\documentclass[10pt,a4paper]{article}
\usepackage[latin1]{inputenc}
\usepackage{amsmath}
\usepackage{amsfonts}
\usepackage{amssymb}
\usepackage{amsthm}
\usepackage{graphicx}
\usepackage{enumitem}
\usepackage{epstopdf}

\usepackage{soul} %fuer text durchstreichen mit \st

\usepackage{float} %for \begin{figure}[H]...\end{figure} inside of minipage (note that the [H] is important ...)

\usepackage[a4paper,margin=2.5cm]{geometry}
\usepackage{caption}
\usepackage{natbib}
\usepackage{aas_macros}

\usepackage{hyperref}

\usepackage{pdflscape}	%Seite im Querformat auch im pdf viewer
\usepackage{siunitx}

\DeclareSIUnit\Msol{M_\odot}
\DeclareSIUnit\pc{pc}
\DeclareSIUnit\kpc{\kilo\pc}
\DeclareSIUnit\yr{yr}
\DeclareSIUnit\Myr{\mega\yr}
\DeclareSIUnit\Gyr{\giga\yr}
\DeclareSIUnit\AU{AU}

\newcommand*\diff{\mathop{}\!\mathrm{d}} %writing now "\diff x" in an Integral gives a good "dx" at the end of the integral
\newcommand*\R{\mathbb R} %domain of real numbers
\newcommand*\measure{\mathcal L}

\newcommand\BT{BT08}

\DeclareMathOperator\supp{supp}

\newtheorem{thm}{Theorem}[section]
\newtheorem{lem}[thm]{Lemma}
\newtheorem{prop}[thm]{Proposition}

\newtheorem{innercustomlem}{Lemma}
\newenvironment{customlem}[1]
{\renewcommand\theinnercustomlem{#1}\innercustomlem}
{\endinnercustomlem}

\theoremstyle{definition}

\theoremstyle{remark}
\newtheorem*{rem*}{Remark}
\newtheorem{rem}[thm]{Remark}

\numberwithin{equation}{section}

\title{\LARGE\textbf{Modelling Spiral Galaxies}}
\author{Joachim Frenkler\\ Fakultät für Mathematik, Physik und Informatik\\Universität Bayreuth \\ D-95440 Bayreuth, Germany \\ joachim.frenkler@uni-bayreuth.de}
\date{\today}

\begin{document}
	\maketitle
	
	\begin{abstract}
		We develop a new technique to equip models of spiral galaxies with self-consistent dynamics that match observations. We apply our technique and construct a model for the Milky Way with a dynamical interstellar medium (ISM). In simulations a four-arm spiral structure emerges from this model that is similar to the one observed in the Milky Way's ISM. Further, in our model the Jeans instability offers an explanation for the observed velocity dispersion of atomic hydrogen in the ISM; this instability vanishes from our model if we choose a velocity dispersion just above the observed one. Our model uses baryonic, dark matter, which resides in the disc and is dynamically cold. This makes our model a typical example for the Bosma effect.
	\end{abstract}
	
	\section{Introduction}
	
	The Milky Way - and spiral galaxies in general - still ask us many a riddle. To three of them we develop new answers here in this paper:
	\begin{enumerate}
		\item[] Where does the four-arm spiral pattern in the Milky Way's ISM originate from \citep{2019ApJ...885..131Reid,2010ApJ...722.1460SteimanCameron}?
		\item[] Why does atomic hydrogen have in most spiral galaxies the same velocity dispersion well above the value expected from thermal considerations \citep{2009AJ....137.4424TamburroLeroyEtAl}?
		\item[] Is a halo of non-baryonic, dark matter necessary to explain the Milky Way's flat circular velocity curve?
	\end{enumerate}
	The key in answering these questions is a model for our galaxy where the interstellar medium (ISM) is equipped with self-consistent dynamics. These dynamics are difficult to model because the ISM's mass is moving on almost circular orbits around the galactic centre. This `almost' makes things difficult; how to implement it in a self-consistent model? But we cannot ignore it because it is important for stability. We solved this mathematical problem and this enables us to shed a new light on the three questions above.
	
	The distribution function in our model is a function of the third component of the angular momentum and the energy only and our technique can easily be extended to construct a self-consistent, multi-phase model for the whole galaxy; a task that was elusive up to now \citep{2020IAUS..353..101Binney}. Our technique is based on a fixed-point like algorithm, which is an improved version of the algorithm from \cite{2015MNRAS.446.3932AndreassonRein}, plus a good understanding of how distribution functions, which match observations, must look like. The resulting distribution function is comparable with the one of a cut-out Mestel disc like it is, e.g., studied in \cite{1976PhDT........29Zang}, \cite{1981seng.proc..111Toomre} or \cite{2019MNRAS.489..116SellwoodCarlberg}. But here in this paper our distribution function is embedded in a realistic model for our galaxy; not in an infinitely extended Mestel disc with infinite mass, like in the papers just cited.
	
	How do we model the Milky Way in this paper? The Milky Way has three baryonic components, a bulge, a stellar disc and the interstellar medium (ISM)\footnote{In fact each of the Milky Way's baryonic components is made up of several distinct sub-components. But for this paper our simplified three component model is sufficient.}. We include the bulge and the stellar disc as rigid components and model the ISM dynamically. In the Milky Way, the ISM's gaseous mass is confined to a very thin disc, which exhibits spiral patterns and where all mass is moving on almost circular orbits around the galactic centre. Simplifying, we assume in our model that the disc is razor-thin, i.e., all mass is restricted to live in the plain, and - at first - we ignore the spiral patterns and assume that the disc is axially symmetric. One could be tempted to simplify further and assume that all mass is on purely circular orbits. However, this would be a bad idea, because such a disc destroys itself very fast \citep[§6.2.3]{2008gady.book.....BinneyTremaine}. It is therefore important that we have at each position a dispersion of the velocities -- and this makes things complicated. In total we will search for a rigid bulge, a rigid stellar disc and an axisymmetric distribution function $f(x,v)\geq 0$ on position-velocity space with $x,v\in\R^2$ that models the ISM.
	
	Since most of the mass in our galaxy is on almost circular orbits, one can deduce from the observational data how the circular velocity curve looks like \citep{2019ApJ...871..120Eilers}. From this, one can calculate the axisymmetric gravitational potential $U_{Gal}$ of our galaxy. In our model we assume that the potential is time-independent and we demand that the distribution function $f$ is time-independent, too. This is the case if $f$ is constant along each particle orbit, i.e., if for a given test particle with orbit $(x(t),v(t))$, where
	\begin{align*}
	\dot{x} &= v, \\
	\dot{v} &= -\nabla U_{Gal}(x),
	\end{align*}
	it holds that
	\begin{equation*}
	\frac{\diff}{\diff t}f(x(t),v(t)) = 0.
	\end{equation*}
	This is the case if and only if $f$ is a solution of the (time-independent) collisionless Boltzmann equation\footnote{In mathematics this equation is often called the Vlasov equation.}
	\begin{equation} \label{Vlasov-Poisson system eq 1}
	v\cdot \partial_x f - \nabla U_{Gal} \cdot \partial_v f = 0.
	\end{equation}
	We use Newton's law of gravitation and thus the gradient of the gravitational potential that corresponds to $f$ is given by
	\begin{equation} \label{Vlasov-Poisson system eq 2}
	\nabla U_f(x) := G\int_{\R^2} \frac{x-y}{|x-y|^3} \Sigma_f(y) \diff y, \quad x\in\R^2,
	\end{equation}
	where
	\begin{equation} \label{Vlasov-Poisson system eq 3}
	\Sigma_f(x) := \int_{\R^2} f(x,v) \diff v
	\end{equation}
	is the (flat) density on position space $\R^2$ that corresponds to $f$, and $G$ is the constant of gravitation. We search for a model for the Milky Way, so we demand that the gravitational potential generated by our model must equal the gravitational potential $U_{Gal}$ of our galaxy, i.e.,
	\begin{equation}  \label{Vlasov-Poisson system eq 4}
	\nabla U_{Gal} = \nabla U_{bulge} + \nabla U_{st.disc} + \nabla U_f.
	\end{equation}
	We call $f$ a self-consistent model for the Milky Way's ISM if it satisfies \eqref{Vlasov-Poisson system eq 1} - \eqref{Vlasov-Poisson system eq 4}.
	
	Mathematically the problem of finding such models is challenging because we have to find models that are compactly supported, fast rotating and reasonable stable. Each of the following models fails at least at one aspect: The infinitely extended Mestel disk mentioned above is not compactly support and has infinite mass. A cold disk where all mass is on purely circular orbits is highly unstable. And while the models from \cite{2015MNRAS.446.3932AndreassonRein} and \cite{2006FirtRein} are compactly supported and stable, they are not fast rotating. Here in this paper we present a new method to construct models that succeed in all three aspects.
	
	The outline of this paper is as follows: In Section \ref{section-Mestel-disc-to-Milky-Way-model} we study a new class of self-consistent models for the Mestel disc, analyse their inner structure and construct from this understanding our model for the Milky Way. We describe how multiphase models can be constructed with our technique. In Section \ref{section-Comparision-with-other-mass-models} we compare our model with other models from the literature and classify our model as a typical example for the Bosma effect. In Section \ref{section-Stability-and-Spiral-Activity} we study numerically the stability of our model. There are two instabilities and they correctly predict the spiral structure in the Milky Way's ISM and the velocity dispersion of atomic hydrogen. This explanations would fail if we included non-baryonic, dark matter. In Section \ref{section-Conclusion} we summarize our results and identify the tasks that should be tackled next.
	
	\section{From the Mestel disc to a realistic model of the Milky Way} \label{section-Mestel-disc-to-Milky-Way-model}
	
	\subsection{The Mestel disc} \label{subsection-Mestel-disc}
	
	As starting point we analyse first the so called Mestel disc \citep{1963MNRAS.126..553Mestel}. This disc has the flat, axially symmetric density
	\begin{equation*}
	\Sigma_0(r) := \frac{v_0^2}{2\pi G r}, \quad r > 0,
	\end{equation*}
	where $v_0>0$ is a constant with dimension of velocity. The derivative of the axisymmetric gravitational potential generated by this flat mass distribution can be approximated using the well known formula for a spherically symmetric mass distribution
	\begin{equation} \label{U0-spherical-approximation}
	U_0' (r) \approx \frac{GM(r)}{r^2} = \frac{v_0^2}{r}
	\end{equation}
	where $M(r)$ denotes the mass inside the radius $r$. In general this is only a rough approximation for a flat mass distribution but in the special case of the Mestel disc, this approximation gives indeed the correct values \cite[§2.6.1a]{2008gady.book.....BinneyTremaine} and we have for every $r>0$
	\begin{equation} \label{U0prime}
	U_0' (r) = \frac{v_0^2}{r}.
	\end{equation}
	The velocity on a circular orbit is related to the derivative of the potential at radius $r$ via the simple formula
	\begin{equation} \label{vc=sqrt(rU'(r))}
	v_c(r) = \sqrt{r U_0'(r)}.
	\end{equation}
	Hence the Mestel disc has an everywhere flat circular velocity curve with $v_c(r)=v_0$. For the understanding of galaxies like the Milky Way, which exhibit an almost flat circular velocity curve, it is therefore useful to analyse first the analytically accessible Mestel disc. In the following theorem we equip this disc with dynamics:
	
	\begin{thm}
		Let $v_0>0$ and let $U_0(r)=v_0^2\log r$, $r>0$, be the potential of the Mestel disc. Take a function $\Phi_0:[0,\infty)\rightarrow[0,\infty)$ measurable such that
		\begin{equation*}
		0 < I := \int_0^\infty \int_{-\infty}^\infty \frac{1}{v_2}\Phi_0\left(\frac{v_1^2+v_2^2}{2} - v_0^2 \log  \frac{v_2}{v_0}  - \frac{v_0^2}{2}\right) \diff v_1 \diff v_2 < \infty.
		\end{equation*}
		Set
		\begin{equation*}
			C_0 := \frac{v_0^2}{2\pi G I}
		\end{equation*}
		and
		\begin{equation} \label{definition-f0}
			f_0(L_z,E) := \frac{C_0}{L_z} \Phi_0\left( E - v_0^2\log\frac{L_z}{v_0} - \frac{v_0^2}{2} \right) \mathbf{1}_{\{L_z>0\}}
		\end{equation}
		where 
		\begin{equation*}
		L_z(x,v) := x_1v_2-x_2v_1
		\end{equation*}
		is the third component of the angular momentum and
		\begin{equation*}
		E(x,v) := \frac{1}{2}|v|^2 + U_0(x)
		\end{equation*}
		is the local energy. This $f_0$ is a self-consistent model for the Mestel disc in the following sense:
		
		The density that belongs to $f_0$ equals the density of the Mestel disc:
		\begin{equation*}
			\Sigma_{f_0} = \Sigma_0.
		\end{equation*}
		Further along every solution of
		\begin{align} \label{equ thm Mestel disc ODE} 
			\dot{x} &= v, \\
			\dot{v} &= -\nabla U_0(x) \nonumber
		\end{align}
		$L_z$ and $E$ are conserved, because $U_0$ is axisymmetric and time-independent. Thus $f_0$ is constant along every solution of the ODE \eqref{equ thm Mestel disc ODE} and in this sense it solves the Vlasov equation \eqref{Vlasov-Poisson system eq 1} with $\nabla U_{Gal}$ replaced by $\nabla U_0$.
	\end{thm}

	\begin{rem*}
		In the following we refer to $L_z$ as the angular momentum because the other two components of the angular momentum are zero.
		
		A sufficient condition for $0<I<\infty$ is for example that $\Phi_0\in L^\infty_+([0,\infty))$, has compact support and does not vanish everywhere.
	\end{rem*}

	\begin{proof}[Proof of the Theorem]
		For $x\in\R^3\backslash\{0\}$ and $r=|x|$ we use the transformation
		\begin{equation*}
			(v_r,v_t) := \left( \frac{x\cdot v}{r}, \frac{L_z}{r} \right).
		\end{equation*}		
		$v_r$ denotes the velocity in radial direction and $v_t$ the velocity in tangential direction. Using this transformation we get
		\begin{align} \label{equ proof thm Mestel disc argument of Phi0}
			E - v_0^2\log \frac {L_z }{v_0} - \frac{v_0}{2} & = \frac { v_r^2+v_t^2}{2} + v_0^2\log r - v_0^2 \log\frac{rv_t}{v_0} - \frac{v_0}{2}  \nonumber\\
			& = \frac{v_r^2+v_t^2}{2} - v_0^2 \log \frac{v_t}{v_0} - \frac{v_0}{2}.
		\end{align}
		Since $\log s \leq s-1$ for all $s>0$,
		\begin{align}
			E - v_0^2\log \frac {L_z }{v_0} - \frac{v_0}{2} & \geq \frac { v_r^2+v_t^2}{2} - v_0(v_t-v_0) - \frac{v_0}{2} \nonumber\\
			&= \frac{v_r^2}{2} + \frac{(v_t-v_0)^2}{2} \geq 0. \label{equ proof thm Mestel disc argument of Phi0 is non-negative}
		\end{align}
		Thus the argument of $\Phi_0$ is everywhere non-negative and $f_0$ is well defined. Further \eqref{equ proof thm Mestel disc argument of Phi0} implies
		\begin{align*}
			\Sigma_{f_0}(r) & = \int f_0(L_z,E) \diff v \\
			& = \frac{C_0}{r} \int_0^\infty\int_{-\infty}^\infty \frac{1}{v_t} \Phi_0\left( \frac{v_r^2+v_t^2}{2} - v_0^2 \log \frac{v_t}{v_0} - \frac{v_0}{2} \right) \diff v_r \diff v_t\\
			& = \frac{v_0^2}{2\pi G r} = \Sigma_0(r).
		\end{align*}
		Thus $f_0$ is a self-consistent model for the Mestel disc in the above sense.
	\end{proof}
	
	There are two choices for $\Phi_0$ where $f_0$ resembles known distributions functions for the Mestel disc: The first choice is to set $\Phi_0$ as the $\delta$-distribution. In this case we have a cold disc were all mass is on purely circular orbits. This is because the argument of $\Phi_0$ is zero if and only if the radial component of the velocity $v_r=0$ and the tangential component of the velocity $v_t=v_0$ (see inequation \eqref{equ proof thm Mestel disc argument of Phi0 is non-negative}). The second choice is $\Phi_0(\eta) = \exp(-C\eta)$ with $C>0$. In this case we get Toomres model for the Mestel disc that was studied extensively in \cite{1976PhDT........29Zang}.
	
	At a first glance $f_0$ might look a bit complicated and unmotivated, but with the help of the next Lemma we are able to write it down in a simpler form that is more intuitive.
	
	\begin{lem}
		Let $v_0>0$ and let $U_0(r)=v_0^2\log r$, $r>0$, be the potential of the Mestel disc. Every orbit in the potential $U_0$ can be characterized uniquely -- up to rotations and shifts in time -- by its values for $L_z$ and $E$. Let us study orbits with $(L_z,E)\in (0,\infty)\times\R$; we call this half-plane the $L_z$-$E$-plane. All these orbits are moving counter-clock wise around the origin. The angular momentum-energy curve of circular orbits
		\begin{align*}
			(L_c(r),E_c(r)) & = \left(rv_0,v_0^2\log r +\frac{v_0^2}{2}\right) \\
			& = \left( L_c, v_0^2 \log \frac{L_c}{v_0} + \frac{v_0^2}{2}\right) = \left(L_c,E_c\left(\frac{L_c}{v_0}\right)\right)
		\end{align*}
		can either be characterized by $r\in(0,\infty)$ or by $L_c\in(0,\infty)$. It divides the $L_z$-$E$-plane into two parts. All admissible orbits have an $L_z$-$E$-coordinate above the angular momentum-curve, there are no orbits below. The orbits that are almost circular are those that are close to the angular momentum curve.
	\end{lem}
	
	\begin{proof}
		We only consider orbits with $L_z\neq 0$ since these are the orbits that do not pass through the origin. Let $L_z\neq0$ and let
		\begin{equation*}
		U_{\text{eff}}(s) := v_0^2\log s + \frac{L_z^2}{2s^2}, \quad s > 0,
		\end{equation*}
		be the corresponding effective potential. We have
		\begin{equation*}
			U_{\text{eff}}'(s) = \frac{1}{s} \left( v_0^2 - \frac{L_z^2}{s^2} \right).
		\end{equation*}
		Hence $U_{\text{eff}}$ is decreasing for $0<s<r_c = L_z/v_0$. It is increasing for $s>r_c$ and it takes its minimum at $s=r_c$. Further
		\begin{equation*}
			\lim_{s\searrow 0} U_{\text{eff}}(s) = \lim_{s\rightarrow \infty} U_{\text{eff}}(s) = \infty.
		\end{equation*}
		Let $(x(t),v(t))$ be an orbit in the potential $U_0$ with angular momentum $L_z\neq 0$. Set $r:=|x|$ and $v_r:=x\cdot v/r$. Then
		\begin{align} \label{equ proof Lz E plane ODE}
			& \dot r = v_r, \\
			& \dot v_r = - U_{\text{eff}}'(r). \nonumber
		\end{align}
		The energy $E$ along this orbit is constant and we have
		\begin{equation*}
			E = \frac 12 |v_r|^2 + U_{\text{eff}}(r) \geq U_{\text{eff}}(r_c).
		\end{equation*}
		Thus there are uniquely determined $0<r_1< r_c < r_2$ such that
		\begin{equation*}
			E = U_{\text{eff}}(r_1) = U_{\text{eff}}(r_2).
		\end{equation*}
		The right side $(v_r,-U'_{\text{eff}}(r))$ of the ODE \eqref{equ proof Lz E plane ODE} is locally Lipschitz continuous for $(r,v_r)\in(0,\infty)\times\R$. Thus uniqueness implies that there are $t_1,t_2\in\R$ such that $r(t_1)=r_1$ and $r(t_2) = r_2$, and that every orbit with the same values $L_z$ and $E$ is identical to $(x(t),v(t))$ up to rotations and shifts in time. If we consider another orbit $(\tilde x, \tilde v)$ with angular momentum $\tilde L_z \neq L_z$ then obviously the orbit is different from $(x,v)$ because at the same position the orbits will have different tangential velocities. If $\tilde L_z = L_z$ but the energy $\tilde E \neq E$, then $\tilde r_1 \neq r_1$ and $\tilde r_2 \neq r_2$ and hence the orbits are different, too; $\tilde r_1$ and $\tilde r_2$ are defined in the same manner as $r_1$ and $r_2$. Thus every orbit in the potential $U_0$ of the Mestel disc is uniquely characterized by its values for $L_z$ and $E$.
		
		Consider now two test particles with the same angular momentum $L_z>0$. Assume that the first particle is moving on a circular orbit ar radius $r_c = L_z/v_0$ and that the second particle is moving on an eccentric orbit. A short look at the effective potential tells us that at some time $t\in\R$ the second particle has to appear at the radius $r_c$, too. Then the energy of the first particle is 
		\begin{equation*}
			E_1 = U_{\text{eff}}(r_c)
		\end{equation*}
		and the energy of the second particle is
		\begin{equation*}
			E_2 = U_{\text{eff}}(r_c) + \frac{1}{2} |v_r(t)|^2 > U_{\text{eff}}(r_c) = E_1.
		\end{equation*}
		Hence the $L_z$-$E$-coordinate of every orbit is located above the angular momentum-energy curve of circular orbits. Further a particle is on an almost circular orbit if the radial component of its velocity is small, i.e., if it has an $L_z$-$E$-coordinate close to the angular momentum-energy curve.
	\end{proof}
	
	Now let us write down $f_0$ a second time:
	\begin{equation} \label{f0=Psi*Phi(E-Ec)}
		f_0(L_z,E) = \frac{C_0}{L_z} \Phi_0 \left( E - E_c\left(\frac{L_z}{v_0}\right) \right)\mathbf{1}_{\{L_z>0\}}.
	\end{equation}
	For the rest of this paper we use the following simple form for $\Phi_0$:
	\begin{equation} \label{phi-definition}
		\Phi_0(\eta) :=
		\begin{cases}
		1 \quad\text{if } \eta < (2\sigma)^2,\\
		0 \quad\text{else,}
		\end{cases} 
	\end{equation}
	where $\sigma>0$ is a parameter with the dimension of velocity. Now we can explain the structure of $f_0$ in a much more intuitive way: We take a narrow stripe along the angular momentum-energy curve of circular orbits in the potential $U_0$, namely of thickness $(2\sigma)^2$, and define $f_0$ on it. This is the part $\Phi_0(E-E_c)$. For self-consistency the density generated by $f_0$ must be $\Sigma_0$, for this purpose we need also the prefactor $C_0/L_z$. Since we want only orbits that rotate counter-clockwise, we exclude all orbits with $L_z<0$.
	
	We are interested in models where all mass is on almost circular orbits. Thus we need a narrow stripe and a small parameter $\sigma$. We see in the next Lemma that the parameter $\sigma$ gives the dispersion of the tangential velocities:
	
	\begin{lem} \label{lemma velocity dispersion of model for Mestel disc} \label{lemma-d-measures-the-scatter-of-vtan}
		When $\sigma\searrow 0$ the average tangential velocity in the model $f_0$ for the Mestel disc is
		\begin{equation*}
			v_{t,avg} = v_0 + o(\sigma)
		\end{equation*}
		independent of radius. The dispersion of the tangential velocities is
		\begin{equation*}
			\sigma + o(\sigma)
		\end{equation*}
		and also independent of radius; in the rest of this paper we refer to $\sigma$ as the velocity dispersion. Further the dispersion of the radial velocities is $\sqrt 2\sigma + o(\sigma)$ and
		\begin{equation} \label{C0approximation}
			C_0 = \frac{v_0^3}{8\sqrt 2 \pi^2 G \sigma^2} + o(\sigma^{-2}).
		\end{equation}
	\end{lem}

	\begin{rem} \label{remark velocity max and min of model for Mestel disc}
		We can calculate the minimal and the maximal appearing tangential velocities $v_{t,min}$ and $v_{t,max}$ in the model $f_0$ explicitly:
		\begin{equation*}
			v_{t,min/max} = v_0 \sqrt{-W_{0/-1}\left[ -\exp\left( -\frac{8\sigma^2}{v_0^2} - 1 \right) \right]}
		\end{equation*}
		where $W_0$ and $W_{-1}$ denote the two real branches of the Lambert W function.
	\end{rem}

	The proofs of Lemma \ref{lemma velocity dispersion of model for Mestel disc} and of Remark \ref{remark velocity max and min of model for Mestel disc} are both a bit lengthy and hence we put them to the appendix.
	
	For the interested reader it may be stated that the two (equivalent) definitions \eqref{definition-f0} and \eqref{f0=Psi*Phi(E-Ec)} of $f_0$ did not appear from nowhere. By studying symmetry properties of orbits one finds the following scaling property: If $X(t)$ is an orbit in the potential $U_0$ of the Mestel disc then $Y(t)=\alpha X(t/\alpha)$ is for every $\alpha>0$ an orbit, too, i.e., both $X(t)$ and $Y(t)$ are solutions of the ODE $\dot x = v$, $\dot v = - \nabla U_0$. Searching for distribution functions that make use of this scaling property soon leads to an ansatz of the form \eqref{f0=Psi*Phi(E-Ec)}.
	
	\subsection{A cut-out Mestel disc resembling the Milky Way's ISM} \label{subsection-cutout-Mestel-disc}
	
	This is the point where we leave the Mestel disc and start to construct from it a self-consistent model with finite mass and extension. It is plausible to assume that the dynamics of a galaxy should be similar to \eqref{f0=Psi*Phi(E-Ec)} in a region where the circular velocity curve is almost flat. The Milky Way has such a flat curve between \SI{5}{\kilo\pc} and \SI{25}{\kilo\pc} from the galactic centre \citep{2019ApJ...871..120Eilers}. Luckily the Milky Way belongs also to the minority of galaxies that have a central depression of their hydrogen distribution, which makes up the bulk mass of the Milky Way's interstellar medium (ISM). Thus for the Milky Way we are in the situation that most mass of the ISM is located in the region where the circular velocity curve is flat. This makes it an ideal candidate to be modelled with a distribution function similar to \eqref{f0=Psi*Phi(E-Ec)}. This is somewhat a fortunate coincidence since on the one hand this is the easiest situation where we can deduce a finitely extended, self-consistent model from \eqref{f0=Psi*Phi(E-Ec)} and on the other hand the ISM is the part of the visible galaxy that asks us most riddles.
	
	To get from the Mestel disc to a model with finite extension, we will now drop several orbits from $f_0$. First we cut a hole into the central region. By choosing only orbits with
	\begin{equation*}
	L_z > v_0R_1
	\end{equation*}
	for some $R_1 > 0$, we drop most orbits that live within the region $0<r<R_1$. Further we want a finitely extended model, so we demand
	\begin{equation*}
	L_z < v_0R_2
	\end{equation*}
	for some $R_2 > R_1$, thus dropping most orbits that live beyond $R_2$. In what follows the cut out central hole will do just fine, but at the border $R_2$ it will we necessary to further cut out every orbit that crosses $R_2$. This is achieved by demanding
	\begin{equation*}
	E < U_0(R_2) + \frac{L_z^2}{2R_2^2}.
	\end{equation*}
	With these three cut-offs most orbits live beyond $R_1$, but there is no orbit beyond $R_2$. Our new distribution function in orbital form reads as follows:
	\begin{align} \label{f1(E,L)}
	f_1(L_z,E) := \frac{C}{L_z} \Phi_0 \left( E - E_c\left(\frac{L_z}{v_0}\right) \right)&\mathbf{1}_{\{v_0R_1<L_z<v_0R_2\}}\\
	\times&\mathbf{1}_{\{E < U_0(R_2) + L_z^2/(2R_2^2)\}} \nonumber
	\end{align}
	where $C>0$ will be determined below.
	
	We want to model the Milky Way's ISM with this distribution function. In view of the Milky Way's hydrogen distribution \citep[Figure 9.19]{1998gaas.book.....BinneyMerrifield} and its circular velocity curve \citep{2019ApJ...871..120Eilers} we choose
	\begin{equation} \label{parameters-R1-R2-v0_and_d-below}
	R_1 = \SI{4}{\kilo\pc}, R_2 = \SI{21}{\kilo\pc} \text{ and } v_0=\SI{230}{\km\per\s}.
	\end{equation}
	For simplicity we had already chosen $\Phi_0=\mathbf{1}_{[0,(2\sigma)^2]}$ in \eqref{phi-definition}. As shown in Lemma \ref{lemma-d-measures-the-scatter-of-vtan}, with this choice the velocity dispersion of the model is equal to $\sigma$. \cite{2008AJ....136.2782Leroy} calculated the velocity dispersion of atomic hydrogen in the outer regions of several  nearby spiral galaxies. Atomic hydrogen is an abundant gas in the ISM which dominates the outer parts of spiral galaxies like the Milky Way. They found that most galaxies have a dispersion of $\SI[separate-uncertainty=true]{11\pm 3}{\km\per\s}$. The concrete value, that we choose for $\sigma$, affects only little the resulting mass model, but it is important for the stability that we study in more detail in Section \ref{section-Stability-and-Spiral-Activity}. There we vary $\sigma$ and look at the different behaviour of the resulting dynamical models. For the present we fix $\sigma=\SI{11}{\km\per\s}$, and thus choose a dispersion in the middle of the measurements of Leroy et al.
	
	Further we will in the following smooth out the integral kernel of the gradient of the gravitational potential to take into account the observed thickness of the ISM's disc. We assume a constant scale height\footnote{According to \cite{2001RvMP...73.1031Ferriere} most gas of the ISM is cold and warm atomic hydrogen with scale heights between \SI{100}{\pc} and \SI{400}{\pc}, followed by molecular hydrogen with scale heights between \SI{120}{\pc} and \SI{140}{\pc}. If we choose another scale height, e.g., $z_g=\SI{100}{\pc}$ this does hardly change the properties of the resulting model.} $z_g = \SI{300}{\pc}$ for the ISM. Nevertheless, we still define the density $\Sigma_f$ on the planar space $\R^2$, but we replace the gradient of the gravitational potential \eqref{Vlasov-Poisson system eq 2} by
	\begin{equation} \label{gradUf-with-scale-height}
	\nabla U_f(x) := G\int_{\R^2} \Sigma_f(y)\frac{x-y}{(|x-y|^2 + \delta_z^2)^{3/2}};
	\end{equation}
	here $\delta_z = 1.5z_g$ is the average distance in $z$-direction when we draw two test particles at random from the spatial density
	\begin{equation*}
	\Sigma_f(x_1,x_2)\exp\left( -|z|/z_g \right).
	\end{equation*}
	We give a proof of the relation $\delta_z=1.5z_g$ in the appendix.
	
	\begin{figure}
		\begin{center}
			\input{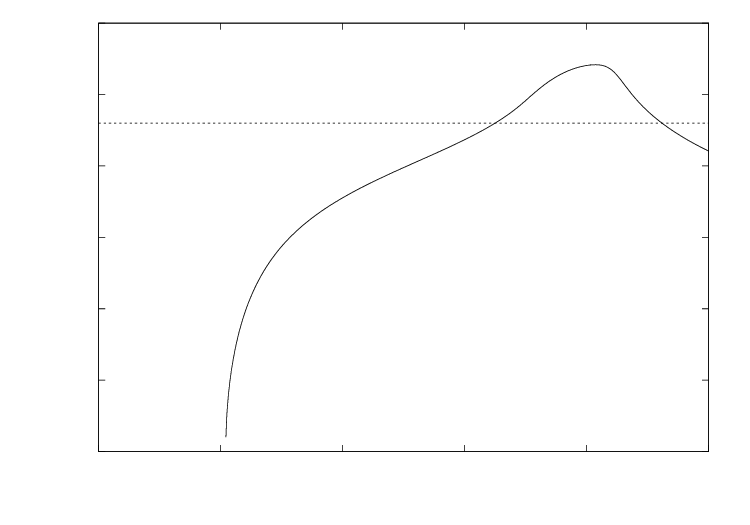}
		\end{center}
		\caption{The circular velocity curve (solid line) that corresponds to the cut-out model $f_1$ if we choose the weight $C\approx C_0$, and the original, constant circular velocity curve (dashed line) that corresponds to the Mestel disc. The circular velocity curve of the cut-out model is no longer flat. To make it flat again, we need to include bulge and stellar disc in our model, too, and reduce the mass of our cut-out Mestel disc, which represents the ISM.}
		\label{figure-Sigma1-and-corresponding-vc}
	\end{figure}
	
	Let us continue with the above parameters. From $f_1(L_z,E)$ we get $f_1(x,v)$ in Cartesian coordinates by replacing
	\begin{equation*}
	E=U_0(r)+\frac{|v|^2}{2}\quad\text{and}\quad L_z=x_1v_2-x_2v_1.
	\end{equation*}
	Then we can calculate numerically the density $\Sigma_1(r)=\int f_1 \diff v$. Further, we calculate the corresponding potential and the circular velocity curve. The circular velocity curve that corresponds to $f_1$ is shown in Figure \ref{figure-Sigma1-and-corresponding-vc}, where we set 
	\begin{equation*}
	C=\SI{6.6e24}{\Msol\per\s}\approx C_0
	\end{equation*}
	according to the approximation \eqref{C0approximation}. Obviously the circular velocity curve is no longer flat. Mainly the force generated by the central mass is missing to support a flat circular velocity curve in the region between $R_1$ and $R_2$. This missing mass has to be 'replaced' by the bulge and the stellar disc which we implement as rigid components. We implement the bulge as spherically symmetric and since it only extends out to approximately $\SI{1.9}{\kpc}<R_1$ \citep[§2.7(a)]{2008gady.book.....BinneyTremaine} its actual shape\footnote{There is evidence that the Milky Way's bulge is in fact bar shaped and not spherical and extends out to radius \SI{3.5}{\kpc} \citep{2002MNRAS.330..591Bissantz}. But as discussed by \cite[§2.7(e)]{2008gady.book.....BinneyTremaine} such a bar would have only a small impact on the dynamics beyond \SI{5}{\kpc} where most of the dynamical mass in our model is located. In this paper we do not discuss this issue in more detail. For the model that we construct and analyse we keep a spherical bulge.} does not affect our model and we take for simplicity
	\begin{equation*}
	\rho_b(x) := A \left(1 - \frac{|x|}{\SI{1.9}{\kpc}}\right) \quad \text{for }x\in\R^3 \text{ and } |x|\leq \SI{1.9}{\kpc}.
	\end{equation*}

	For the stellar disc we assume a scale length $R_d =\SI{3.2}{\kpc}$ \citep[compare Model 2 from ][§2.7]{2008gady.book.....BinneyTremaine} and define the density
	\begin{equation*}
	\Sigma_d(r) = B \exp\left(-\frac{r}{R_d}\right)\quad\text{for }r>0.
	\end{equation*}
	For the stellar disc we assume a disc thickness of $\SI{500}{\pc}$\footnote{\cite{2008gady.book.....BinneyTremaine} included two stellar discs that have similar properties. One has a scale height \SI{300}{\pc} and the other \SI{1000}{pc}. We consider one disc with average parameters and set the scale height to \SI{500}{pc}} and smooth out the gradient of its gravitational potential as in \eqref{gradUf-with-scale-height}.
	
	We have now three components, namely bulge, stellar disc and ISM with the free parameters $A,B,C>0$. We fit these parts together such that our model reproduces the observed circular velocity curve of the Milky Way as closely as possible. This curve was measured in \cite{2019ApJ...871..120Eilers} and we refer to it as $v_{c,MW}$. We choose $A_1,B_1,C_1>0$ such that
	\begin{equation} \label{minimize-vc2-vb2-vd2-vg2}
	\int_{\SI{5}{\kpc}}^{\SI{25}{\kpc}} \left( v_{c,MW}^2 - v_{b,1}^2 - v_{d,1}^2 - v_{g,1}^2 \right)^2 \diff r
	\end{equation}
	becomes minimal; here $v_{b,1}$, $v_{d,1}$ and $v_{g,1}$ denote the circular velocity curve of the bulge, the stellar disc and the gaseous component ISM respectively where we have replaced $A,B,C$ by $A_1,B_1,C_1$. The integral borders \SI{5}{\kpc} and \SI{25}{\kpc} are the lower and upper border of the range covered by \citeauthor{2019ApJ...871..120Eilers} and we treat $v_{c,MW}$ as a piece-wise linear function. Calculating numerically the optimal parameters we get
	\begin{align*}
	A_1 &= \SI{2.7}{\Msol\per\cubic\pc} \\
	B_1 &= \SI{1300}{\Msol\per\square\pc} \\
	C_1 &= \SI{1.7e24}{\Msol\per\s} \approx 0.26 C_0.
	\end{align*}
	With these parameters fixed we can calculate the circular velocity curve of our model
	\begin{equation*}
	v_{c,1} = \sqrt{v_{b,1}^2 + v_{d,1}^2 + v_{g,1}^2}.
	\end{equation*}
	In Figure \ref{figure-vc1} both $v_{c,1}$ and the measured curve of the Milky Way $v_{c,MW}$ are shown. As can be seen, this is already quite a good fit. Nevertheless the model is not self-consistent yet, because the potential that belongs to the mass model is different from the potential that we assumed for the dynamics.
	
	\begin{figure}
		\begin{center}
			\input{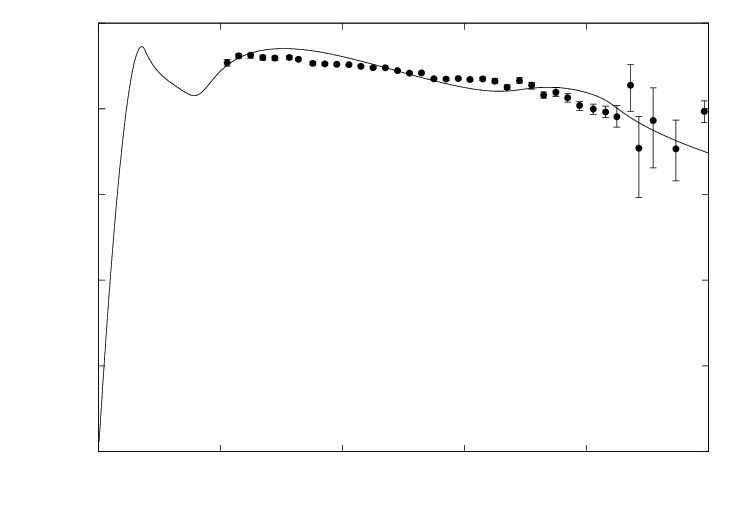}
		\end{center}
		\caption{The resulting circular velocity curve (solid line) after we have also included a bulge and a stellar disc and after we have weighted the three baryonic components optimally. In this model the ISM is represented by $f_1$ with $C =C_1\approx 0.26 C_0$. Its dynamics are not self-consistent yet. The dots mark the circular velocity curve of the Milky Way as measured by \protect\cite{2019ApJ...871..120Eilers}, which we try to approximate with our model.} 
		\label{figure-vc1}
	\end{figure}

	At this point one might object that there is no (non-baryonic) dark matter in our model. We did not forget the dark matter, we omitted it intentionally. One should not be upset about this, one should rather be surprised that what we have done is possible at all: Instead of introducing dark matter, we used higher values for $A_1$, $B_1$ and $C_1$ than current models do. As you can see in Figure \ref{figure-vc1} this way, too, it is possible to explain the Milky Way's circular velocity curve. We discuss the physical implications of this  model without non-baryonic dark matter in detail in the Sections \ref{section-Comparision-with-other-mass-models} and \ref{section-Stability-and-Spiral-Activity}. But for the moment I ask you that we postpone these discussions there, so that first we can finish the mathematics: How do we have to modify the model further so that the dynamics become self-consistent?

	\subsection{A Model of the Milky Way with self-consistent dynamics for the ISM} \label{subsection-algorithm}
	
	Recall how in the previous section the dynamical part of our model, the ISM, was constructed. The ISM is located mostly between $R_1$ and $R_2$ where the circular velocity curve is almost flat. We thought about how a distribution function in such a region should look like and in \eqref{f1(E,L)} we defined $f_1$ under the assumption of a logarithmic potential that gives rise to an exactly flat rotation curve everywhere. After adding a bulge and a stellar disc to the model, the resulting circular velocity curve has now some bumps and is slightly decaying in the relevant region between $R_1$ and $R_2$ (Figure \ref{figure-vc1}), but it is almost flat. So the initially assumed logarithmic potential is close to the resulting potential, and thus $f_1$ is also close to a self-consistent model. We want to iterate what we have done so far and use the following algorithm to construct a self-consistent model:
	
	\filbreak
	\paragraph*{Algorithm}
	for the construction of a model of the Milky Way where the ISM is equipped with self-consistent dynamics
	\begin{enumerate}
		\item Given a gravitational potential $U_i(r)$. Calculate 
		\begin{align*}
		v_{c,i}(r) &= \sqrt{r U_i'(r)},\\
		L_{c,i}(r) &= rv_{c,i}(r), \\
		E_{c,i}(r) &= U_i(r) + v_{c,i}^2(r)/2
		\end{align*}
		for $R_1 < r < R_2$
		\item Choose $R_2'<R_2$ maximal such that $L_{c,i}$ is strictly increasing on $[R_1,R_2']$ \footnote{This is necessary because the disc of the ISM is truncated at $R_2$. As a result $v_{c,i}$ is decaying very rapidly near $R_2$. This is an effect due to the flatness of the disc. As a result $L_{c,i}$ is not monotonous in this region.} and define the inverse map of $L_{c,i}(r)$:
		\begin{equation*}
		[L_{c,i}(R_1), L_{c,i}(R_2')] \ni L_z \mapsto r_c(L_z)
		\end{equation*}
		(hence $r_c(L_z)$ is the radius where a test particle with angular momentum $L_z$ is on a circular orbit)
		\item \label{algorithm-step-definition-fi} Define
		\begin{align*}
		f_{i+1}(L_z,E) := \frac{C}{L_z} \Phi_0 \left( E - E_c(r_c(L_z)) \right)&\mathbf{1}_{\{L_{c,i}(R_1)<L_z<L_{c,i}(R_2')\}} \\
		\times&\mathbf{1}_{\{E < U_0(R_2') + L_z^2/(2R_2'^2)\}} \nonumber
		\end{align*}
		\item Replace
		\begin{equation*}
		E = U_i(x) + \frac{|v|^2}{2} \text{ and } L_z = x_1v_2-x_2v_1
		\end{equation*}
		and calculate the flat density $\Sigma_{i+1} = \int f_{i+1} \diff v$ and the corresponding potential and circular velocity curve $v_{g,i+1}$
		\item Replace $A,B,C$ by $A_{i+1},B_{i+1},C_{i+1}$ and choose them such that
		\begin{equation*}
		\int_{5kpc}^{25kpc} \left( v_{c,MW}^2 - v_{b,i+1}^2 - v_{d,i+1}^2 - v_{g,i+1}^2 \right)^2 \diff r
		\end{equation*}
		is minimal
		\item Calculate the total potential $U_{i+1}(r)$ of all three baryonic components and return to the first step
		
	\end{enumerate}
	
	To measure the convergence of our algorithm we look at
	\begin{equation*}
	\delta_i := \frac{||\Sigma_{i+1} - \Sigma_i||_2}{||\Sigma_i||_2}
	\end{equation*}
	where $||\cdot||_2$ denotes the $L^2$-Norm on $\R^2$. With the parameters chosen in \eqref{parameters-R1-R2-v0_and_d-below}, $\delta_1\approx 0.10$, $\delta_i$ decreases in each iteration step roughly by a factor between 0.5 and 0.7, and we stop the algorithm after twelve iterations when $\delta_i<0.001$. The resulting distribution function $f$ is a self-consistent model for the Milky Way's ISM where the bulge and the stellar disc are rigid components. Or more precisely, it is as close to a self-consistent model as possible: We cannot distinguish it anymore from a self-consistent model on a computer.
	
	\begin{figure}
		\begin{center}
			\input{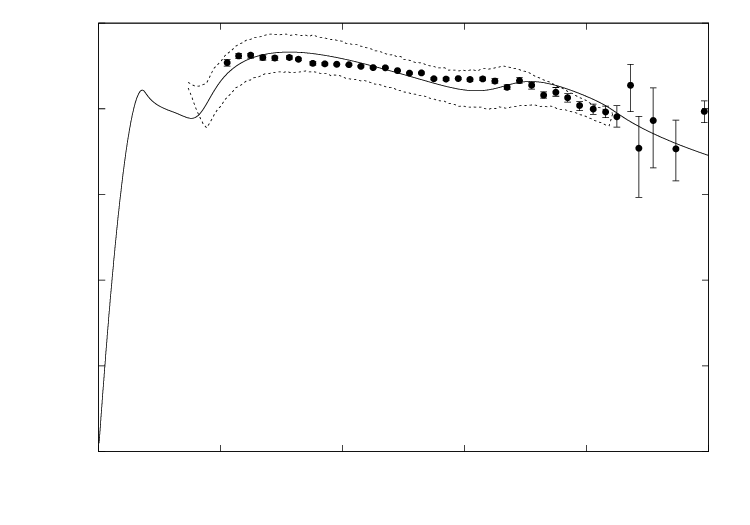}
		\end{center}
		\caption{The circular velocity curve (solid line) of our final model. For comparison we show also the circular velocity curve of the Milky Way (thick dots) measured by \protect\cite{2019ApJ...871..120Eilers}. The velocity dispersion (dashed lines) is shown. In our model now the ISM is equipped with self-consistent dynamics.}
		\label{figure-vcFinal-and-vtanTube-and-vcMW}
	\end{figure}
	
	The circular velocity curve and the velocity dispersion in our model are shown in Figure \ref{figure-vcFinal-and-vtanTube-and-vcMW}. The circular velocity curve and the low velocity dispersion resemble quite well the properties of the Milky Way. This is very nice, but with this observation alone we can not be satisfied yet. We have to cover two more very relevant topics: As already mentioned we have to discuss the lack of non-baryonic dark matter in our model (Section \ref{section-Comparision-with-other-mass-models}). Further we have to analyse the stability of our model (Section \ref{section-Stability-and-Spiral-Activity})? These two topics are closely coupled.
	
	\subsection{A self-consistent, multiphase model of the entire galaxy} \label{subsection-multiphase-models}
	
	Before in the next sections we we study the physical implications of our model, let us take a closer look on the algorithm itself. The algorithm of the previous section is very powerful since it is highly customizable. It can easily be extended to create a multiphase, self-consistent model of the entire galaxy. This is a task that was elusive up to now; see, e.g., the review from \cite{2020IAUS..353..101Binney}.
	
	How can the algorithm from Section \ref{subsection-algorithm} be extended to construct a multiphase, self-consistent model for the entire galaxy? Look for this on Step \ref{algorithm-step-definition-fi} in the algorithm where we define our distribution function. In this definition we have a prefactor $C/L_z$. This prefactor was motivated from the Mestel disc where it was necessary for self-consistency (see equation \eqref{f0=Psi*Phi(E-Ec)}). But in our finitely extended model this is no longer necessary. There the algorithm takes care that the dynamics become self-consistent. So we can replace $C/L_z$ by any suitable function $\psi(L_z)$. Further we can add to our model as many different distribution functions as we want and choose for every one a different prefactor $\psi$. Each of these new distribution functions has to be updated in Step \ref{algorithm-step-definition-fi} of the algorithm.
	
	It could for example be convenient to decompose the ISM into atomic hydrogen HI and molecular hydrogen H2. This can easily be achieved by including two distribution function $f_{HI}$ and $f_{H2}$ with suitable prefactors $\psi_{HI}$ and $\psi_{H2}$.
	
	In the same way we can equip a stellar disc with dynamics. But it is important to note that a stellar disc must be defined also in regions where the circular velocity curve is not flat but rising. So one has to study first how a suitable ansatz in these regions looks like.
	
	It is also possible to add a three-dimensional, dynamical dark matter halo with self-consistent dynamics to the model (adding a rigid halo is trivial). The idea behind this is as follows: First we have to estimate the densities the halo and the disc components shall have in the final model. Then we can approximate the potential of this (rigid) galaxy model with a spherical potential using equation \eqref{U0-spherical-approximation} for the disk components. Using Eddington inversion \citep{2008gady.book.....BinneyTremaine} we determine a distribution function $f_{h}(E)$ for the halo. If the approximation \eqref{U0-spherical-approximation} had been correct, the dynamics of this halo would be self-consistent. However the approximation \eqref{U0-spherical-approximation} is in general not exact. But with the algorithm from the previous section we can nevertheless make the dynamics of the halo self-consistent easily. All we have to do is to replace the energy in Step 4 of the algorithm by its three-dimensional analogon
	\begin{equation*}
		E(x,z) = U_i(x,z) + |v|^2
	\end{equation*}
	where $U_i(x,z)$, $x\in\R^2$, $z\in\R$, is the three dimensional potential that is created by the halo together with the disk components. It is not necessary to modify the halo distribution function in the Steps 1, 2 and 3. In Step 5 the halo gets weighted as the other galaxy components. This modified algorithm will converge to a model where all components have self-consistent dynamics. Since the approximation \eqref{U0-spherical-approximation} is in general not bad, the density of the resulting halo will be close to our initial estimate.
	
	\section{Comparing our mass model with observational data} \label{section-Comparision-with-other-mass-models}
	
	In the introduction we posed the question:
	\begin{enumerate}
		\item[] Is a halo of non-baryonic, dark matter necessary to explain the Milky Way's flat circular velocity curve?
	\end{enumerate}
	In the previous section we have constructed a model that explains the Milky Way's flat circular velocity curve out to \SI{25}{\kpc}. The disc of this model has an extension of only \SI{21}{\kpc}, there is no large halo of non-baryonic, dark matter involved and the densities of bulge, stellar disc and ISM conform to observations up to a prefactor. Nevertheless, to explain the circular velocity curve a certain amount of mass is needed that creates the necessary gravitational potential. Since most current models make use of a dark matter halo, while our model does not, our model needs more baryonic mass than current models and a mass gap occurs. In this section we study how large this mass gap is by comparing our model with other models from the literature.

	\subsection{Comparing our mass model with the one in \protect\cite{2008gady.book.....BinneyTremaine}} \label{subsection-comparision-with-BT}
	
	\begin{table}
		\begin{center}
			\begin{tabular}{lccc}
				\hline
				& \BT & Here & Factor \\
				& $10^{10}\,\si{\Msol}$ & $10^{10}\,\si{\Msol}$&\\
				\hline
				Total	& 21.5 & 15.0 & 0.7 \\
				Total baryonic 	& 4.3 & 15.0 & 3.5 \\
				\hline
				Bulge			& 0.36 & 1.4 & 3.9 \\
				Stellar disc	& 3.0 & 8.5 & 2.8 \\
				ISM 			& 1.0 & 5.1 & 5.1 \\
				Dark matter		& 17.1 & - & - \\
				\hline
			\end{tabular}
		\end{center}
		\caption{Masses of the different components of the Milky Way. The first column contains the masses of Model 2 in Binney and Tremaine (\BT), the second one the masses of our self-consistent model (Here), and the third column (Factor) states by how much the mass of \BT\ must be multiplied to match the mass of our model. Only masses within \SI{25}{\kpc} were included. In our model all mass lies within \SI{25}{\kpc}. In \BT\ a high amount of mass, in particular dark matter, lies beyond \SI{25}{\kpc}; this mass we ignore in our comparison. (The masses from \BT\ were calculated using central surface densities of $\SI{463}{\Msol\per\square\pc}$ and $\SI{73}{\Msol\per\square\pc}$ for the stellar disc and the ISM respectively)}
		\label{table-Mass-in-BT-and-Here}
	\end{table}
	
	In §2.7 of \cite{2008gady.book.....BinneyTremaine} (hereafter referenced as \BT) two mass models for the Milky Way were constructed. Their Model 2 assumes a stellar disc with a scale length $R_d = \SI{3.2}{\kpc}$ as we did in Section \ref{section-Mestel-disc-to-Milky-Way-model}. Further they also included a bulge with similar properties as our bulge and their model for the ISM is based on the observational data from \cite{1998gaas.book.....BinneyMerrifield} that we used too. Therefore we compare Model 2 from \BT\ with our model.
	
	The masses of the Milky Way's components in the model of \BT\ and in our model are listed in Table \ref{table-Mass-in-BT-and-Here}. All mass in our model is confined to a disc with radius \SI{21}{\kpc} and we explain with this mass the circular velocity curve of the Milky Way out to \SI{25}{\kpc}. In the model of \BT\ much mass (mostly dark matter) lies beyond the edge of the visible galaxy. To compare the two models properly, we list therefore in the table only masses within a ball with radius \SI{25}{\kpc}.
	
	We see that the total mass of our model and the one of the model in \BT\ take similar values. This is to be expected because both models must explain the same circular velocity curve and for this similar amounts of mass are necessary. Nevertheless, the total mass of our model is about 30 per cent lower than the mass in \BT. Since in the model of \BT\ only one quarter of the mass is baryonic while in our model all mass is baryonic, necessarily a mass gap arises. Our model needs 3.5 times as much baryonic mass as the model in \BT. The missing mass distributes almost uniformly over the three baryonic components of our model. We have to multiply the masses of the three baryonic components in our model by factors between 2.8 and 5.1 to reproduce the masses in \BT.

	\subsection{Our model for the Milky Way as an example for the Bosma effect} \label{subsection-Bosma}
	
	That it is possible to explain circular velocity curves of spiral galaxies by scaling the gaseous content of these galaxies was already noticed by \cite{1981AJ.....86.1825Bosma}. When he measured the densities and dynamics in the outer parts of spiral galaxies, he calculated the disc density necessary to explain the observed circular velocity curve and compared it to the observed density of the gas. He noticed that as a rule the ratio of the two is roughly constant in the outer parts of the galaxies in his sample. This phenomenon is called the Bosma effect. \cite{2011A&A...532A.121Hessman} used this phenomenon to explain the circular velocity curves of 17 galaxies from The Nearby HI Galaxy Survey (THINGS) without invoking non-baryonic, dark matter. They scaled the observed densities of both the stellar and the gaseous discs like we did for the Milky Way in Sections \ref{subsection-cutout-Mestel-disc} and \ref{subsection-algorithm}. They interpreted the scaled stellar and the scaled gaseous discs as proxies for other -- presumedly non-stellar -- mass components that reside in the disc and have not been observed, yet. They found good agreement between observed and predicted circular velocity curves. We have summarized their scaling factors for the stellar and the gaseous discs\footnote{In \cite{2011A&A...532A.121Hessman} the scaling factors for the stellar discs are given by $1+f_{\text{disc}}$ and the scaling factors for the gaseous discs by $1+f_{\text{HI}}$. The values for $f_{\text{disc}}$ and $f_{\text{HI}}$ are tabulated in their paper. Since the gaseous discs of \citeauthor{2011A&A...532A.121Hessman} contained only atomic hydrogen HI, \citeauthor{2011A&A...532A.121Hessman} multiplied their gaseous discs with an additional factor of 1.39 to correct the disc densities for the presence of Helium and heavier elements. We have not invoked such an additional factor since the model for the ISM from  \cite{2008gady.book.....BinneyTremaine} includes already similar corrections.} in Figure \ref{figure-histogram-factors}.
	
	\begin{figure}
		\begin{center}
			\include{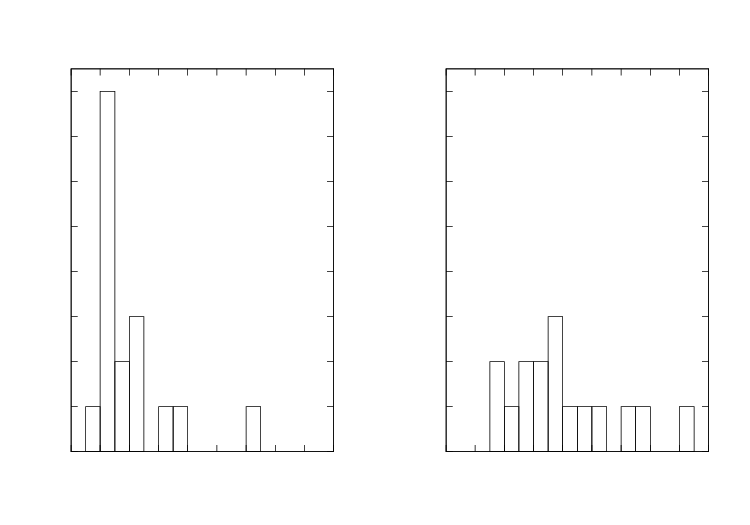}
		\end{center}
		\caption{Scaling factors used by \cite{2011A&A...532A.121Hessman} for the stellar and the HI discs of 17 spiral galaxies from the THINGS sample. One outlier (scaling factor 28.0 for the HI content) is not displayed in the right histogram. With the such scaled discs \citeauthor{2011A&A...532A.121Hessman} can explain the circular velocity curves of the respective galaxies without invoking non-baryonic, dark matter. Our factors 2.8 for the stellar and 5.1 for the gaseous disc of the Milky Way fall in the midst of their factors.}
		\label{figure-histogram-factors}
	\end{figure}
	
	Both \BT\ and we (Section \ref{subsection-cutout-Mestel-disc}) took a functional form for the ISM that is similar to the observed density of atomic plus molecular hydrogen reported by \cite{1998gaas.book.....BinneyMerrifield}. Further both \BT\ and we took the same functional form for the stellar disc. In our model the stellar disc is 2.8 times  and the ISM 5.1 times more massive than the corresponding discs from \BT. These two factors fall in the midst of the factors from \cite{2011A&A...532A.121Hessman} making our model for the Milky Way a typical example for the Bosma effect.
	
	\subsection{Higher densities of the ISM measured by the Voyager probes} \label{subsection-voyager-probes}
	
	An argument that the density of the ISM could indeed be higher than currently assumed is provided by the Voyager 1 and 2 probes which in 2012 and 2018 left the heliosphere and entered the interstellar medium. They are the first artificial objects to do so. Inside the heliosphere the electron density is very low (about 0.001 particles per \si{\cubic\cm}). In the interstellar medium current models predict a higher value of about 0.04 particles per \si{\cubic\cm} \citep{1993Sci...262..199GurnettKurthEtAl}. Measurements carried out by the space probes motivated \citeauthor{1993Sci...262..199GurnettKurthEtAl} already in 1993 to postulate that there must be a 'pill up' region in front of the heliospheric nose where the electron density is higher than the predicted value. The Voyager 1 and 2 probes entered the ISM far off the region where the pile up was expected to be \citep{2020ApJ...900L...1KurthGurnett}. First they measured an electron density of \SI{0.04}{\per\cubic\cm} and \SI{0.05}{\per\cubic\cm} close to the estimate mentioned above. But after travelling \SI{20}{\AU} more this density rose to \SI{0.13}{\per\cubic\cm} and \SI{0.12}{\per\cubic\cm}. Roughly three times higher than expected. Since the two space probes entered the interstellar medium at different positions, Gurnett and Kurth expect that this high electron density is a large scale feature that can be found everywhere in the direction of the heliospheric nose.
	
	Gurnett and Kurth discussed some possible explanations for this high density but concluded that the question of its origin cannot be answered satisfactory. We would like to add another possible explanation to their list: Could it be that this high electron density is not just a local phenomenon near the heliospheric nose, but that is is real? Meaning that the electron density is indeed higher than expected everywhere in the Milky Way and that this points toward a higher density of the whole interstellar medium, consistent with our model?
	
	\section{Stability, Spiral Structure and Velocity Dispersion} \label{section-Stability-and-Spiral-Activity}
	
	Let us pose the question: Is our model stable? The answer to this question is: No, it is unstable. It can suffer from two instabilities. And this good. Because these instabilities take care that our model offers simple answers to the two questions from the introduction:
	\begin{samepage}
		\begin{enumerate}
			\item[] Where does the four-arm spiral pattern\footnote{Due to our position inside the Milky Way's disk it is difficult to determine the exact structure of the Milky Way's spiral arms (see, e.g., the discussion in the introduction of \cite{2021A&A...651A.104Poggio}). However, at the moment a four-arm spiral structure is the most accurate visualization of what the Milky Way looks like \citep{2020RAA....20..159ShenZheng,2019ApJ...885..131Reid,2010ApJ...722.1460SteimanCameron}.} in the Milky Way's ISM originate from?
			\item[] Why does atomic hydrogen have in most spiral galaxies the same velocity dispersion well above the value expected from thermal considerations?
		\end{enumerate}
	\end{samepage}
	For these answers it is important that the ISM mass in our model is as high as we have seen in the previous section. Models that use an ISM disc with a lower mass embedded in a halo of non-baryonic, dark matter cannot answer these questions as easily as our model does (Section \ref{subsection-DM-cannot-explain-spirals-and-dispersion}).
	
	In this section we show several simulations where the initial particles were drawn at random from the distribution function constructed in Section \ref{subsection-algorithm} and where the equations of motion were integrated numerically. In Section \ref{subsection-spiral-activity}, where we study the spiral activity, we choose in our model $\sigma=\SI{11}{\km\per\s}$ and in Section \ref{subsection-Jeans-instability}, where we study the Jeans instability, we look on a model with $\sigma=\SI{9}{\km\per\s}$. Details on our numerical methods can be found in the appendix.
	
	\subsection{Spiral structure in our model and in the Milky Way} \label{subsection-spiral-activity}
	
	\begin{figure}
		\begin{center}
			\include{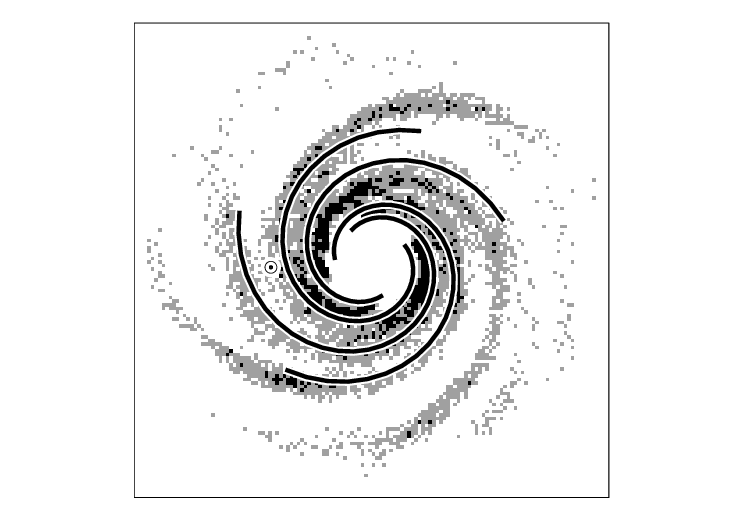}
		\end{center}
		\caption{Spiral structure (Background image) that has formed in our model within \SI{400}{\Myr} out of the initially axially symmetric disc, overlaid with the four spiral arms that were calculated from observational data by \protect\cite{2010ApJ...722.1460SteimanCameron} (solid lines); the position of the sun is marked by $\odot$. Given that the spiral arms in our model form spontaneously, the similarity between them and the observed arms is astonishing. The plate covers \SI{40}{\kpc} x \SI{40}{\kpc}. In the background image black corresponds to densities above \SI{120}{\Msol\per\square\pc}, gray to densities above \SI{60}{\Msol\per\square\pc}, and white to densities below.}
		\label{figure-spiral-arms-Milky-Way}
	\end{figure}
	
	\begin{figure}
		\begin{center}
			\input{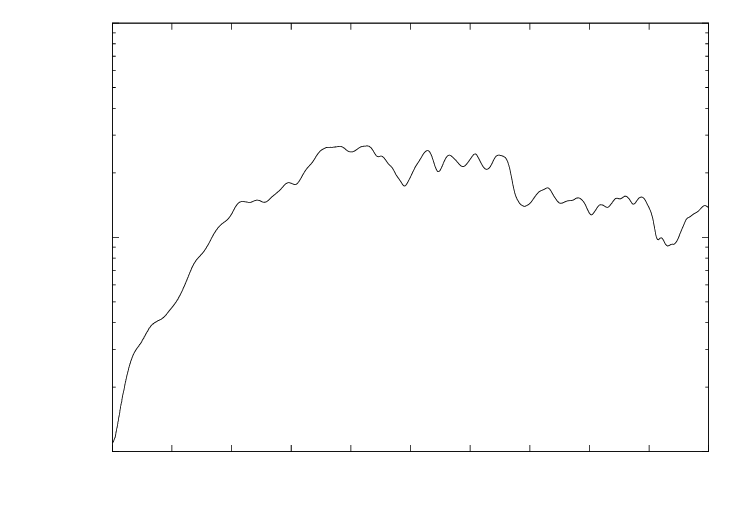}
		\end{center}
		\caption{The RMS of the tangential accelerations is very low initially but grows exponentially for about \SI{400}{\Myr} (take note of the logarithmic scale for the y-axis). It still continues to grow for \SI{300}{\Myr} more at a slower pace and then it decays gradually. The growth of the tangential accelerations corresponds to local overdensities which become denser and denser. These overdensities result in a spiral structure that resembles very well the observed spiral structure of the Milky Way's ISM (see Figure \ref{figure-spiral-arms-Milky-Way}).}
		\label{figure-RMS-of-atan}
	\end{figure}
	
	The first instability our model suffers from offers a simple explanation for the observed four-arm spiral pattern in the Milky Way's ISM (see Figure \ref{figure-spiral-arms-Milky-Way}). To understand what happens there, let us take a look on the tangential accelerations. In axial symmetry, these accelerations would be zero. However, in our simulation the particles were drawn at random from the distribution function and hence these accelerations are different from zero although they are very small initially. If now we run the simulation, these tangential accelerations grow exponentially. This can be seen very well in Figure \ref{figure-RMS-of-atan} where we have plotted the root mean square (RMS) of the tangential accelerations as a function of time:
	\begin{equation*}
	RMS(a_{tan}) = \sqrt{\sum_i \left(\frac{x_{i,1} a_{i,2} - x_{i,2} a_{i,1}}{r_i}\right)^2},
	\end{equation*}
	where $a = \dot{v}$ and the sum ranges over all particles in the simulation. These growing tangential accelerations correspond to local overdensities which become denser and denser and result in a spiral structure with four large spiral arms that match the observed spiral arms in the Milky Way's ISM to a high degree; in Figure \ref{figure-spiral-arms-Milky-Way} we have overlaid the spiral structure of our model with the four-arm spiral structure that was observed by \cite{2010ApJ...722.1460SteimanCameron}. Given that the spiral structure in our model forms spontaneously, the similarity is astonishing.

	\citeauthor{2010ApJ...722.1460SteimanCameron} discuss several possible explanations why the Milky Way's ISM has a four-arm spiral structure. Our model offers a further explanation: Assuming the mass of the ISM is as high as in our model, then spiral activity is self-excited, it is independent from the dynamical properties of the rest of the galaxy and it gives rise to the spiral pattern that is observed in the Milky Way.
	
	\subsection{Reliability of our simulation} \label{subsection-cloud-cloud-collsions}
	
	\begin{figure}
		\begin{center}
			\input{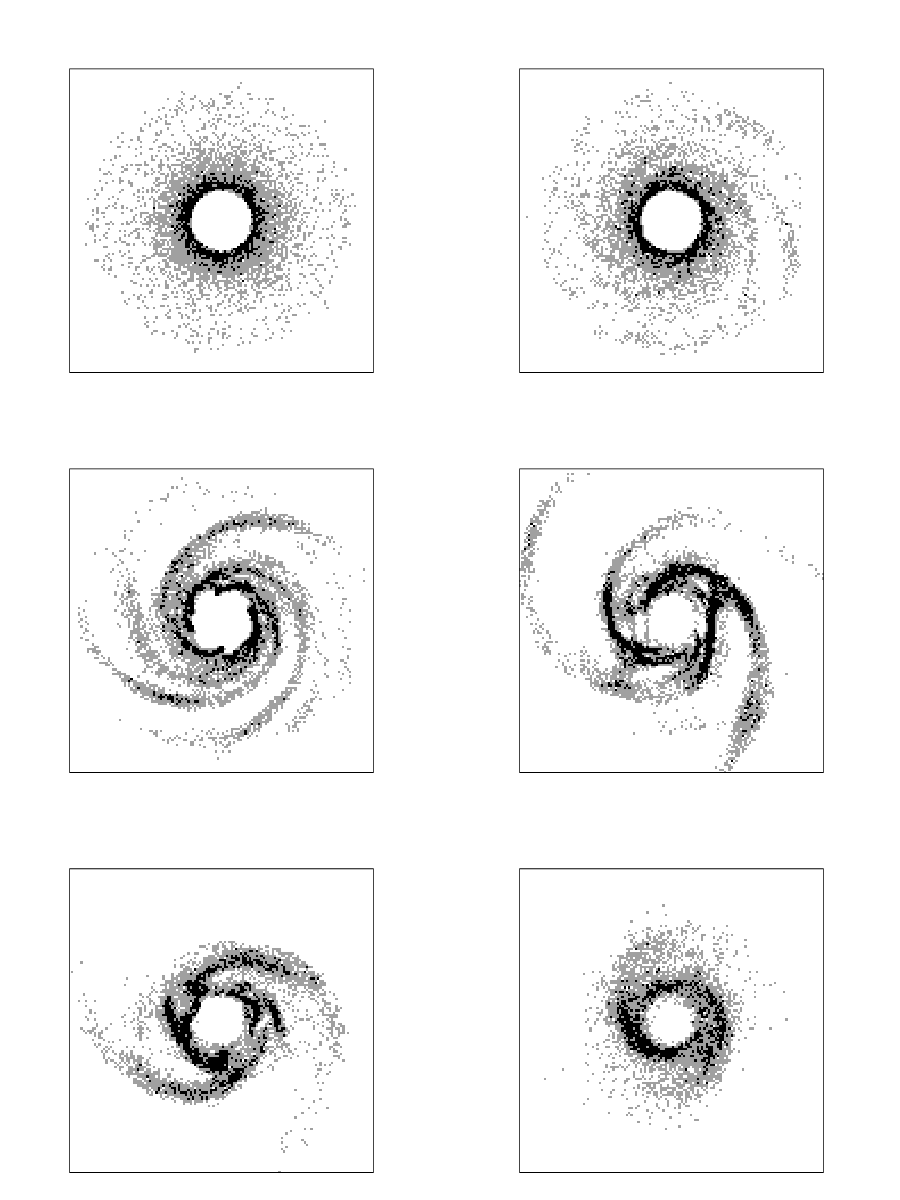}
		\end{center}
		\caption{A simulation of our self-consistent model where the initial particles are drawn at random. After about \SI{200}{\Myr} we can make out a faint multi-arm spiral structure. These spiral arms merge and after \SI{400}{\Myr} form four large spiral arms, which resemble the Milky Way's spiral arms very well. As time continues these arms move, become elongated due to the differential rotation and disrupt and merge with each other. After \SI{800}{\Myr} we still see a three-arm, and after \SI{1.2}{\Gyr} a two-arm spiral structure. Afterwards the spiral activity calms down. After \SI{3}{\Gyr} still a weak, bi-symmetric structure is visible. As the simulation continues, the system converges more and more to a new axially symmetric state. That the spiral activity vanishes during the simulation is due to the increasing velocity dispersion reported in Figure \ref{figure-vtan-400My-1400My}. Each plate covers \SI{40}{\kpc} x \SI{40}{\kpc} and the colour scheme is the same as in Figure \ref{figure-spiral-arms-Milky-Way}. A more detailed video of this simulation can be found on the project homepage: \url{https://www.diffgleichg.uni-bayreuth.de/en/research/spiral-galaxies/index.html}}
		%black: above 20 particles per cell, gray: between 10 and 20 particles, white: below 10 particles
		%disc mass = 5.054e+10 Solar Masses
		%plotted particles 100.000 = 1e5
		%each particle has mass 5.54e5 Sol Masses
		%size of cell: 60kpc/200 = 0.3kpc = 300pc
		%10 particles per cell is density 10*5.54e5/300pc^2 = 61.55 Sol Mass/pc^2
	\end{figure}
	In our simulation the spiral activity does not last forever (Figure \ref{figure-plates-for-simulation}). The reason for this is the velocity dispersion that starts to rise at \SI{400}{\Myr} (Figure \ref{figure-vtan-400My-1400My}). This rising velocity dispersion stabilizes the disc, weakens the spiral activity and is responsible that the simulation converges to a new axially symmetric state.
	
	\begin{figure}
		\begin{center}
			\input{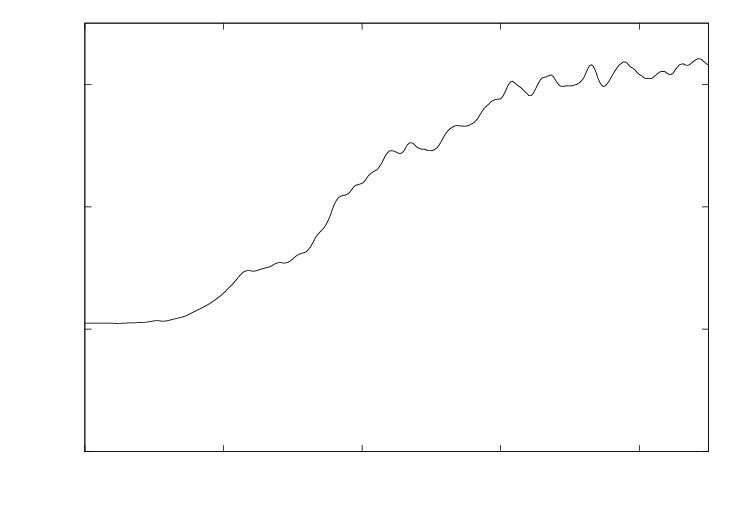}
		\end{center}
		\caption{The time evolution of the mean velocity dispersion in our model between radius \SI{5}{\kpc} and \SI{15}{\kpc}. For the first \SI{400}{\Myr} the velocity dispersion stays at $\sigma = \SI{11}{\km\per\s}$. Between \SI{400}{\Myr} and \SI{1200}{\Myr} the velocity dispersion rises at a constant rate. Afterwards it settles down around \SI{32}{\km\per\s}. Thus it is justified to neglect collisional effects in our simulation till \SI{400}{\Myr}. This is the moment when the spiral structure that we analysed in Figure \ref{figure-spiral-arms-Milky-Way} has manifested.}
		\label{figure-vtan-400My-1400My}
	\end{figure}
	
	In the Milky Way this seems not to happen - and in other spiral galaxies neither. For comparison the Milky Way's stellar disc is assumed to be \SI[separate-uncertainty=true]{8.8\pm 1.7}{\Gyr} old \citep{2005A&A...440.1153DelPelosoEtAl}. Nevertheless, we can observe nearly everywhere in the Milky Way's ISM the same low dispersion of velocities \citep{2017A&A...607A.106Marasco} and also the spiral activity seems to last forever. Why?
	
	In the real galaxy the velocity dispersion of the gas becomes permanently reduced. The Milky Way's gaseous mass in the ISM is not distributed homogeneously but it is concentrated in large clouds. As long as these clouds would move on circular orbits, everything would be fine. But when they become deviated and the velocity dispersion rises, the orbits of these clouds intersect and they collide. In these collisions they loose the radial component of their velocity and continue on on circular orbits. Thus the velocity dispersion decreases again. This keeps the velocity dispersion low and enables a long lasting spiral activity  \citep[§6.1.]{2021arXiv211005615SellwoodMasters}.
	
	We did not model collisional effects in our simulation. Thus the results from our simulation can only be transferred to the real galaxy as long as the velocity dispersion does not rise and collisional effects can be neglected. As can be seen from Figure \ref{figure-vtan-400My-1400My} this is the case for the first \SI{400}{\Myr}. Thus our simulation is reliable until the moment when the four-arm spiral structure, which we discussed in Figure \ref{figure-spiral-arms-Milky-Way}, has manifested.
	
	\subsection{Velocity dispersion and the Jeans instability} \label{subsection-Jeans-instability}
	
	\begin{figure}
		\begin{center}
			\input{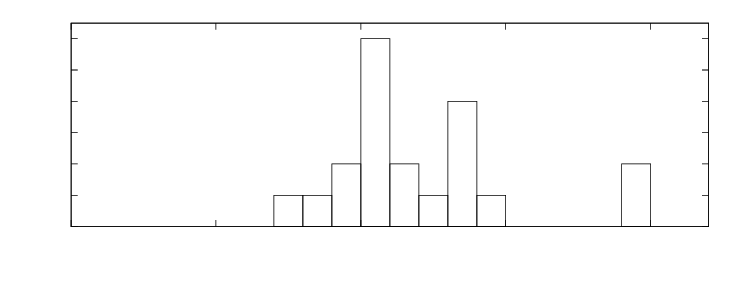}
		\end{center}
		\caption{Histogram of the velocity dispersion of atomic hydrogen in the outer parts of the galactic discs of 20 nearby spiral galaxies. The sample is taken from \protect\cite{2008AJ....136.2782Leroy} and only galaxies with an inclination below 60° were included. Above an inclination of 60° the calculated velocity dispersion is affected by projection errors. We see that most spiral galaxies have a velocity dispersion around \SI{10}{\km\per\s} or \SI{11}{\km\per\s} and there is a gap between dispersion zero and the observed velocity dispersions.}
		\label{figure-histogram-dispersion}
	\end{figure}
	
	The dissipative process just described must be strong because everywhere in the Milky Way atomic hydrogen has a low velocity dispersion slightly below \SI{10}{\km\per\s} (see, e.g., \cite{2017A&A...607A.106Marasco} where they have measured the velocity dispersion in the inner regions of our galaxy). So in the Milky Way the spiral activity does not manage to increase this dispersion like in our simulation. The observed velocity dispersion of atomic hydrogen in the Milky Way is typical for other spiral galaxies, too. In the sample of twenty nearby spiral galaxies from \cite{2008AJ....136.2782Leroy} most spiral galaxies have a dispersion close to \SI{10}{\km\per\s} (see Figure \ref{figure-histogram-dispersion}). This is higher than the value one would expect from thermal considerations \citep{2009AJ....137.4424TamburroLeroyEtAl}. In spiral galaxies most hydrogen can be found either in a cold ($\sim \SI{100}{\K}$) or in a warm ($\sim \SI{8000}{\K}$) thermal equilibrium. Cold atomic hydrogen has a line width of $\sim \SI{1}{km\per\s}$, while warm atomic hydrogen has a line width of $\sim \SI{8}{\km\per\s}$. But if the dissipative process of cloud-cloud collisions is strong, why does it not reduce the velocity dispersion to the minimal thermal value somewhere between \SI{1}{\km\per\s} and \SI{8}{\km\per\s}?\begin{figure}
		\begin{center}
			\input{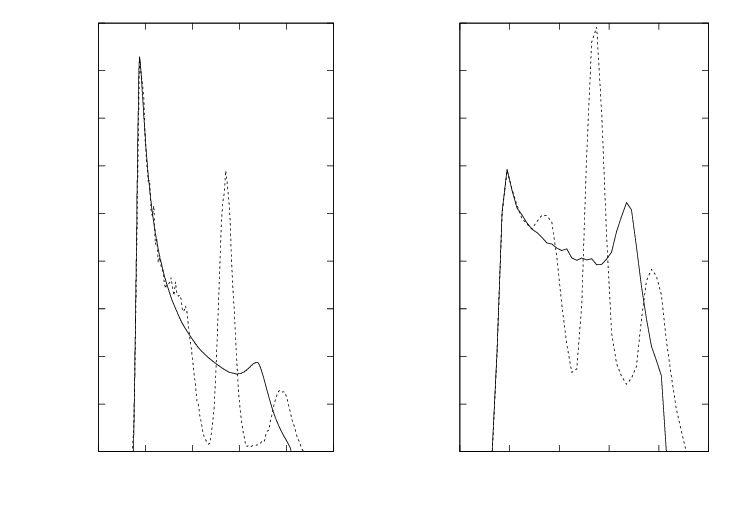}
			\caption{The two figures show the Jeans instability in action if we choose in our model $\sigma=\SI{9}{\km\per\s}$. The left plot shows the density as a function of the radius at time zero (solid line) and after \SI{800}{\Myr} (dashed line). The right plot shows the velocity dispersion as a function of the radius at time zero (solid line) and after \SI{800}{\Myr} (dashed line). After \SI{800}{\Myr} the Jeans instability has rearranged the mass beyond \SI{10}{\kpc}. This mass is now concentrated in two rings around the galactic centre. The velocity dispersion in these rings is higher than initially. In the simulation that was run to create these plots axial symmetry was enforced to suppress the formation of spiral arms so that we can study the Jeans instability isolated. If in this simulation we choose $\sigma=\SI{11}{\km\per\s}$, like in Section \ref{subsection-spiral-activity} above, the model just remains stable.}
			\label{figure-Jeans-Dispersion-9}
		\end{center}
	\end{figure}
	
	Our model offers an answer to this question, again in the form of an instability. If in our model we choose $\sigma \gtrapprox \SI{10}{\km\per\s}$ only the instability from Section \ref{subsection-spiral-activity}, which causes the spiral arms, is active. If, however, we choose $\sigma \lessapprox \SI{10}{\km\per\s}$ a second instability enters the model: The Jeans instability. This instability rearranges the masses and increases the velocity dispersion. In Figure \ref{figure-Jeans-Dispersion-9} we show this instability in action in our model with $\sigma=\SI{9}{\km\per\s}$. The threshold between stability and the Jeans instability coincides with the observed velocity dispersion of atomic hydrogen in the Milky Way and in most other spiral galaxies. A dissipative process -- however strong it may by -- cannot reduce the velocity dispersion below this threshold  because if it does so, the Jeans instability starts to work against it. Thus -- assuming that other spiral galaxies are comparable to the Milky Way -- the Jeans instability can explain why most spiral galaxies share the same velocity dispersion of atomic hydrogen well above the value expected from thermal considerations.
	
	Interestingly, the Jeans instability in our model triggers only in the outskirts of the galactic disc. And it is good that it does not trigger in the inner regions, too, because there the density of the stellar disc is higher than the density of the ISM disc. But at the present the dynamics of the stellar component are missing in our model (see the discussion in Section \ref{subsection-multiphase-models} how we plan to include the dynamics of this component in future models). Obviously these dynamics will affect the Jeans instability in the inner regions of our model. On the contrary in the outskirts of our model the ISM disc, which we have modelled dynamically, dominates the mass. So there our model has to predict the velocity dispersion correctly; and this it does.
	
	\subsection{What about non-baryonic, dark matter} \label{subsection-DM-cannot-explain-spirals-and-dispersion}
	
	Let us pose a last question here in this paper: Do models that use a halo of non-baryonic, dark matter have similar easy explanations for the spiral structure in the Milky Way's ISM and the velocity dispersion of atomic hydrogen? The answers is No, at least not ad hoc. The problem is that a dark matter halo provides too much stability.
	
	In our baryonic model we can freeze four fifths of the ISM mass and study only the dynamical rest. Then the frozen mass is kind of a rigid dark matter component and the remaining dynamical disc resembles an ISM disc like it is predicted by Model 2 of \cite{2008gady.book.....BinneyTremaine} (see Section \ref{subsection-comparision-with-BT}). In such a disc neither the instability that causes the spiral structure (Section \ref{subsection-spiral-activity}) nor the Jeans instability, which explains the velocity dispersion (Section \ref{subsection-Jeans-instability}), is active. So models that make use of dark matter have to search for more complicated answers to explain the two dynamical phenomenons spiral activity and velocity dispersion.

	\subsection{Spiral Arms and Bar-Shaped Bulges result from the same instability} \label{subsection-spirals-and-bulges}
	
	We want to close this section with a short note, which is independent of the previous discussions. When we started this work, we had at first self-consistent models, too, but these lacked many properties of the models presented here in this paper. Almost all of these first models degenerated into bars or lop-sided discs. In this context it is noteworthy that the formation of a bar, of a lop-sided disc or of large scale spiral structures is always preceded by an exponential growth of the tangential accelerations as shown in Figure \ref{figure-RMS-of-atan}. So these three phenomena are all due to the same instability. This instability always is triggered if sufficient mass is in sufficiently rotational motion. But to which result this instability leads, depends on the concrete distribution of the mass and its dynamical properties. We have not examined this instability any further, but it is obvious that a good understanding of it would prove very useful since it is both responsible for the formation of large scale spiral structures and relevant for how bulges are shaped.
	
	\section{Conclusion} \label{section-Conclusion}
	
	We presented a new technique to set up models of spiral galaxies where the dynamics of the different components are self-consistent (Section \ref{section-Mestel-disc-to-Milky-Way-model}). The dynamics of the resulting models  resemble very well the dynamics observed in real galaxies. We applied our technique and constructed a model for the Milky Way where the ISM is equipped with self-consistent dynamics (Sections \ref{subsection-cutout-Mestel-disc} and \ref{subsection-algorithm}). 
	
	We examined the physical properties of our model. Below we summarize the answers this model gives to the three questions from the introduction:
	%Let us summarize the answers our model for the Milky Way gives to the three questions from the introduction:
	\begin{enumerate}
		\item[] Where does the four-arm spiral pattern in the Milky Way's ISM originate from?
	\end{enumerate}
	In our model the origin of this spiral pattern lies in the dynamical properties of the ISM itself. Our self-consistent, axisymmetric model for the ISM suffers from an instability that transforms the ISM disc into a disc with four large spiral arms that resemble very well the spiral arms observed in the Milky Way's ISM (Section \ref{subsection-spiral-activity}).
	
	\begin{enumerate}
		\item[] Why does atomic hydrogen have in most spiral galaxies the same velocity dispersion well above the value expected from thermal considerations?
	\end{enumerate}
	In the outer regions of the galactic disc the Jeans instability is active in our model if we choose a velocity dispersion that is below \SI{10}{\km\per\s} (Section \ref{subsection-Jeans-instability}). Thus if we would include in our model the dissipative process of cloud-cloud collisions (Section \ref{subsection-cloud-cloud-collsions}), this dissipative process can only reduce the velocity dispersion to the threshold between stability and the Jeans instability. Then the Jeans instability starts to work against it and stops the further reduction of the velocity dispersion. This threshold between stability and the Jeans instability coincides with the observed values of the velocity dispersion of atomic hydrogen in the outer regions of spiral galaxies (Figure \ref{figure-Jeans-Dispersion-9}). Assuming that the structure of most spiral galaxies is comparable to the Milky Way's, the Jeans instability offers an explanation why in most spiral galaxies the velocity dispersion of atomic hydrogen gets reduced to the same value somewhere around \SI{10}{\km\per\s}, which is well above the value expected from thermal considerations.
	
	\begin{enumerate}
		\item[] Is a halo of non-baryonic, dark matter necessary to explain the Milky Way's flat circular velocity curve?
	\end{enumerate}
	Our model explains the Milky Way's flat circular velocity curve out to \SI{25}{\kpc} without relying on non-baryonic, dark matter (Section \ref{subsection-algorithm}). Our model has an extension of \SI{21}{\kpc} and is made up only of baryonic matter. The densities of the three baryonic components bulge, stellar disc and ISM match the densities derived from observations up to a prefactor. The three prefactors take values between three and five making our model a typical example for the Bosma effect (Section \ref{subsection-Bosma}). Following the interpretation of the Bosma effect from \cite{2011A&A...532A.121Hessman}, a baryonic, yet unobserved mass component that resides in the disc and traces the known baryonic components can explain the Milky Way's flat circular velocity curve. Then a halo of non-baryonic dark matter is no longer necessary.
	An argument in favour of an ISM density as high as in our model is provided by the measurements of the two voyager probes, the first artificial objects to reach the ISM. Both probes measure a three times higher density than was expected a priori (Section \ref{subsection-voyager-probes}). But perhaps an even stronger evidence for the existence of a baryonic, yet unobserved matter component, which shares the same dynamical properties as the HI gas, is given by the answers to the other two questions: Both answers would fail if we included a stabilizing, non-baryonic dark matter component and reduced the ISM mass, so that it matches current assumptions (Section \ref{subsection-DM-cannot-explain-spirals-and-dispersion}). 
	
	Summarizing the properties of our dynamical model match well to observational data available for the real galaxy. This motivates us to conclude that one should consider the Bosma effect as a promising alternative to the concept of non-baryonic dark matter.
	
	What should be done next? It is necessary to examine also models including a dynamical stellar disk. For this purpose our technique, which is currently limited to regions where the circular velocity curve is flat, must be developed further. The key that enabled us to construct the here presented model was the scaling property of orbits reported at the end of Section \ref{subsection-Mestel-disc}. In every region where the circular velocity curve of a galaxy can be approximated by a power-law\footnote{Approximating a circular velocity curve with a power-law means that there are $C>0$ and $\alpha\in\R$ such that $v_c(r) \approx Cr^\alpha$.} similar scaling properties exist \citep{1998MNRAS.300...83EvansRead}. Using these scaling properties will enable us to model dynamics also in the central regions of spiral galaxies where the circular velocity curve is typically rising linearly. This will enable us to add also a dynamical stellar disc to the model. It is important to check which answers to the three questions from the introduction continue to hold once we include a stellar disc.
	
	Next it would be interesting to extend our research to a larger sample of galaxies. For the THINGS galaxies \citep{2008AJ....136.2563Walter} good observational data is available. This data already enabled several authors to construct mass models for these galaxies that explain the respective circular velocity curves by either using non-baryonic, dark matter \citep{2008AJ....136.2648deBlok}, modified gravity laws (MOND) \citep{2011A&A...527A..76GentileFamaeyBlok} or the Bosma effect \citep{2011A&A...532A.121Hessman}. But these authors did not examine the stability properties of their models using dynamical models. It would be very interesting to examine dynamical versions of these models and ask similar questions as we did here in this paper. Will these dynamical models give similar answers as our model for the Milky Way gave?
	
	And one last question deserves our attention: What is the very nature of the instability that causes the spiral arms to form? In Figure \ref{figure-RMS-of-atan} we have seen that this instability shows up as an exponential growth of the forces in tangential direction. It must be possible to track down the nature of this instability with rigorous mathematical methods. Since this instability is not only responsible for the formation of spiral structures but also for the formation of bar-shaped bulges (Section \ref{subsection-spirals-and-bulges}), a better understanding of this instability would enhance our understanding of both: Spirals in the disc and bulges in the centre of spiral galaxies.
	
	\section*{Acknowledgement}
	
	Funded by the Deutsche Forschungsgemeinschaft (DFG, German Research Foundation) -- Projektnummer RE 885/4-1.
	
	\bibliographystyle{mnras}
	\bibliography{bibliography_math,bibliography_phys,bibliography_mond}

\begin{thebibliography}{}
\makeatletter
\relax
\def\mn@urlcharsother{\let\do\@makeother \do\$\do\&\do\#\do\^\do\_\do\%\do\~}
\def\mn@doi{\begingroup\mn@urlcharsother \@ifnextchar [ {\mn@doi@}
  {\mn@doi@[]}}
\def\mn@doi@[#1]#2{\def\@tempa{#1}\ifx\@tempa\@empty \href
  {http://dx.doi.org/#2} {doi:#2}\else \href {http://dx.doi.org/#2} {#1}\fi
  \endgroup}
\def\mn@eprint#1#2{\mn@eprint@#1:#2::\@nil}
\def\mn@eprint@arXiv#1{\href {http://arxiv.org/abs/#1} {{\tt arXiv:#1}}}
\def\mn@eprint@dblp#1{\href {http://dblp.uni-trier.de/rec/bibtex/#1.xml}
  {dblp:#1}}
\def\mn@eprint@#1:#2:#3:#4\@nil{\def\@tempa {#1}\def\@tempb {#2}\def\@tempc
  {#3}\ifx \@tempc \@empty \let \@tempc \@tempb \let \@tempb \@tempa \fi \ifx
  \@tempb \@empty \def\@tempb {arXiv}\fi \@ifundefined
  {mn@eprint@\@tempb}{\@tempb:\@tempc}{\expandafter \expandafter \csname
  mn@eprint@\@tempb\endcsname \expandafter{\@tempc}}}

\bibitem[\protect\citeauthoryear{{Andr{\'e}asson} \& {Rein}}{{Andr{\'e}asson}
  \& {Rein}}{2015}]{2015MNRAS.446.3932AndreassonRein}
{Andr{\'e}asson} H.,  {Rein} G.,  2015, \mn@doi [MNRAS]
  {10.1093/mnras/stu2346}, \href
  {https://ui.adsabs.harvard.edu/abs/2015MNRAS.446.3932A} {446, 3932}

\bibitem[\protect\citeauthoryear{{Binney}}{{Binney}}{2020}]{2020IAUS..353..101Binney}
{Binney} J.,  2020, in {Valluri} M.,  {Sellwood} J.~A.,  eds,  Vol. 353,
  Galactic Dynamics in the Era of Large Surveys. pp 101--108,
  \mn@doi{10.1017/S1743921319008214}

\bibitem[\protect\citeauthoryear{{Binney} \& {Merrifield}}{{Binney} \&
  {Merrifield}}{1998}]{1998gaas.book.....BinneyMerrifield}
{Binney} J.,  {Merrifield} M.,  1998, {Galactic Astronomy}.
Princeton University Press, Princeton, NJ

\bibitem[\protect\citeauthoryear{{Binney} \& {Tremaine}}{{Binney} \&
  {Tremaine}}{2008}]{2008gady.book.....BinneyTremaine}
{Binney} J.,  {Tremaine} S.,  2008, {Galactic Dynamics}, 2nd edn.
Princeton University Press, Princeton, NJ

\bibitem[\protect\citeauthoryear{{Bissantz} \& {Gerhard}}{{Bissantz} \&
  {Gerhard}}{2002}]{2002MNRAS.330..591Bissantz}
{Bissantz} N.,  {Gerhard} O.,  2002, \mn@doi [\mnras]
  {10.1046/j.1365-8711.2002.05116.x}, \href
  {https://ui.adsabs.harvard.edu/abs/2002MNRAS.330..591B} {330, 591}

\bibitem[\protect\citeauthoryear{{Bosma}}{{Bosma}}{1981}]{1981AJ.....86.1825Bosma}
{Bosma} A.,  1981, \mn@doi [\aj] {10.1086/113063}, \href
  {https://ui.adsabs.harvard.edu/abs/1981AJ.....86.1825B} {86, 1825}

\bibitem[\protect\citeauthoryear{{Eilers}, {Hogg}, {Rix}  \& {Ness}}{{Eilers}
  et~al.}{2019}]{2019ApJ...871..120Eilers}
{Eilers} A.-C.,  {Hogg} D.~W.,  {Rix} H.-W.,   {Ness} M.~K.,  2019, \mn@doi
  [ApJ] {10.3847/1538-4357/aaf648}, \href
  {https://ui.adsabs.harvard.edu/abs/2019ApJ...871..120E} {871, 120}

\bibitem[\protect\citeauthoryear{{Evans} \& {Read}}{{Evans} \&
  {Read}}{1998}]{1998MNRAS.300...83EvansRead}
{Evans} N.~W.,  {Read} J.~C.~A.,  1998, \mn@doi [\mnras]
  {10.1046/j.1365-8711.1998.01863.x}, \href
  {https://ui.adsabs.harvard.edu/abs/1998MNRAS.300...83E} {300, 83}

\bibitem[\protect\citeauthoryear{{Ferri{\`e}re}}{{Ferri{\`e}re}}{2001}]{2001RvMP...73.1031Ferriere}
{Ferri{\`e}re} K.~M.,  2001, \mn@doi [Reviews of Modern Physics]
  {10.1103/RevModPhys.73.1031}, \href
  {https://ui.adsabs.harvard.edu/abs/2001RvMP...73.1031F} {73, 1031}

\bibitem[\protect\citeauthoryear{Firt \& Rein}{Firt \&
  Rein}{2006}]{2006FirtRein}
Firt R.,  Rein G.,  2006, \mn@doi [Analysis] {doi:10.1524/anly.2006.26.4.507},
  \href {https://ui.adsabs.harvard.edu/abs/2006math.ph...5070F/abstract} {26,
  507}

\bibitem[\protect\citeauthoryear{{Gentile}, {Famaey}  \& {de Blok}}{{Gentile}
  et~al.}{2011}]{2011A&A...527A..76GentileFamaeyBlok}
{Gentile} G.,  {Famaey} B.,   {de Blok} W.~J.~G.,  2011, \mn@doi [\aap]
  {10.1051/0004-6361/201015283}, \href
  {https://ui.adsabs.harvard.edu/abs/2011A&A...527A..76G} {527, A76}

\bibitem[\protect\citeauthoryear{{Gurnett}, {Kurth}, {Allendorf}  \&
  {Poynter}}{{Gurnett} et~al.}{1993}]{1993Sci...262..199GurnettKurthEtAl}
{Gurnett} D.~A.,  {Kurth} W.~S.,  {Allendorf} S.~C.,   {Poynter} R.~L.,  1993,
  \mn@doi [Science] {10.1126/science.262.5131.199}, \href
  {https://ui.adsabs.harvard.edu/abs/1993Sci...262..199G} {262, 199}

\bibitem[\protect\citeauthoryear{{Hessman} \& {Ziebart}}{{Hessman} \&
  {Ziebart}}{2011}]{2011A&A...532A.121Hessman}
{Hessman} F.~V.,  {Ziebart} M.,  2011, \mn@doi [\aap]
  {10.1051/0004-6361/201117199}, \href
  {https://ui.adsabs.harvard.edu/abs/2011A&A...532A.121H} {532, A121}

\bibitem[\protect\citeauthoryear{{Kurth} \& {Gurnett}}{{Kurth} \&
  {Gurnett}}{2020}]{2020ApJ...900L...1KurthGurnett}
{Kurth} W.~S.,  {Gurnett} D.~A.,  2020, \mn@doi [ApJL]
  {10.3847/2041-8213/abae58}, \href
  {https://ui.adsabs.harvard.edu/abs/2020ApJ...900L...1K} {900, L1}

\bibitem[\protect\citeauthoryear{{Leroy}, {Walter}, {Brinks}, {Bigiel}, {de
  Blok}, {Madore}  \& {Thornley}}{{Leroy}
  et~al.}{2008}]{2008AJ....136.2782Leroy}
{Leroy} A.~K.,  {Walter} F.,  {Brinks} E.,  {Bigiel} F.,  {de Blok} W.~J.~G.,
  {Madore} B.,   {Thornley} M.~D.,  2008, \mn@doi [AJ]
  {10.1088/0004-6256/136/6/2782}, \href
  {https://ui.adsabs.harvard.edu/abs/2008AJ....136.2782L} {136, 2782}

\bibitem[\protect\citeauthoryear{{Marasco}, {Fraternali}, {van der Hulst}  \&
  {Oosterloo}}{{Marasco} et~al.}{2017}]{2017A&A...607A.106Marasco}
{Marasco} A.,  {Fraternali} F.,  {van der Hulst} J.~M.,   {Oosterloo} T.,
  2017, \mn@doi [A\&A] {10.1051/0004-6361/201731054}, \href
  {https://ui.adsabs.harvard.edu/abs/2017A&A...607A.106M} {607, A106}

\bibitem[\protect\citeauthoryear{{Mestel}}{{Mestel}}{1963}]{1963MNRAS.126..553Mestel}
{Mestel} L.,  1963, \mn@doi [MNRAS] {10.1093/mnras/126.6.553}, \href
  {https://ui.adsabs.harvard.edu/abs/1963MNRAS.126..553M} {126, 553}

\bibitem[\protect\citeauthoryear{{Poggio} et~al.,}{{Poggio}
  et~al.}{2021}]{2021A&A...651A.104Poggio}
{Poggio} E.,  et~al., 2021, \mn@doi [\aap] {10.1051/0004-6361/202140687}, \href
  {https://ui.adsabs.harvard.edu/abs/2021A&A...651A.104P} {651, A104}

\bibitem[\protect\citeauthoryear{{Reid} et~al.,}{{Reid}
  et~al.}{2019}]{2019ApJ...885..131Reid}
{Reid} M.~J.,  et~al., 2019, \mn@doi [\apj] {10.3847/1538-4357/ab4a11}, \href
  {https://ui.adsabs.harvard.edu/abs/2019ApJ...885..131R} {885, 131}

\bibitem[\protect\citeauthoryear{{Sellwood} \& {Carlberg}}{{Sellwood} \&
  {Carlberg}}{2019}]{2019MNRAS.489..116SellwoodCarlberg}
{Sellwood} J.~A.,  {Carlberg} R.~G.,  2019, \mn@doi [MNRAS]
  {10.1093/mnras/stz2132}, \href
  {https://ui.adsabs.harvard.edu/abs/2019MNRAS.489..116S} {489, 116}

\bibitem[\protect\citeauthoryear{{Sellwood} \& {Masters}}{{Sellwood} \&
  {Masters}}{2021}]{2021arXiv211005615SellwoodMasters}
{Sellwood} J.~A.,  {Masters} K.~L.,  2021, arXiv e-prints, \href
  {https://ui.adsabs.harvard.edu/abs/2021arXiv211005615S} {p. arXiv:2110.05615}

\bibitem[\protect\citeauthoryear{{Shen} \& {Zheng}}{{Shen} \&
  {Zheng}}{2020}]{2020RAA....20..159ShenZheng}
{Shen} J.,  {Zheng} X.-W.,  2020, \mn@doi [Research in Astronomy and
  Astrophysics] {10.1088/1674-4527/20/10/159}, \href
  {https://ui.adsabs.harvard.edu/abs/2020RAA....20..159S} {20, 159}

\bibitem[\protect\citeauthoryear{{Steiman-Cameron}, {Wolfire}  \&
  {Hollenbach}}{{Steiman-Cameron}
  et~al.}{2010}]{2010ApJ...722.1460SteimanCameron}
{Steiman-Cameron} T.~Y.,  {Wolfire} M.,   {Hollenbach} D.,  2010, \mn@doi [ApJ]
  {10.1088/0004-637X/722/2/1460}, \href
  {https://ui.adsabs.harvard.edu/abs/2010ApJ...722.1460S} {722, 1460}

\bibitem[\protect\citeauthoryear{{Tamburro}, {Rix}, {Leroy}, {Mac Low},
  {Walter}, {Kennicutt}, {Brinks}  \& {de Blok}}{{Tamburro}
  et~al.}{2009}]{2009AJ....137.4424TamburroLeroyEtAl}
{Tamburro} D.,  {Rix} H.~W.,  {Leroy} A.~K.,  {Mac Low} M.~M.,  {Walter} F.,
  {Kennicutt} R.~C.,  {Brinks} E.,   {de Blok} W.~J.~G.,  2009, \mn@doi [\aj]
  {10.1088/0004-6256/137/5/4424}, \href
  {https://ui.adsabs.harvard.edu/abs/2009AJ....137.4424T} {137, 4424}

\bibitem[\protect\citeauthoryear{{Toomre}}{{Toomre}}{1981}]{1981seng.proc..111Toomre}
{Toomre} A.,  1981, in {Fall} S.~M.,  {Lynden-Bell} D.,  eds, Structure and
  Evolution of Normal Galaxies. pp 111--136, \url
  {https://ui.adsabs.harvard.edu/abs/1981seng.proc..111T}

\bibitem[\protect\citeauthoryear{{Walter}, {Brinks}, {de Blok}, {Bigiel},
  {Kennicutt}, {Thornley}  \& {Leroy}}{{Walter}
  et~al.}{2008}]{2008AJ....136.2563Walter}
{Walter} F.,  {Brinks} E.,  {de Blok} W.~J.~G.,  {Bigiel} F.,  {Kennicutt}
  Robert~C. J.,  {Thornley} M.~D.,   {Leroy} A.,  2008, \mn@doi [\aj]
  {10.1088/0004-6256/136/6/2563}, \href
  {https://ui.adsabs.harvard.edu/abs/2008AJ....136.2563W} {136, 2563}

\bibitem[\protect\citeauthoryear{{Zang}}{{Zang}}{1976}]{1976PhDT........29Zang}
{Zang} T.~A.,  1976, PhD thesis, Massachusetts Institute of Technology, United
  States, \url {https://dspace.mit.edu/handle/1721.1/27444}

\bibitem[\protect\citeauthoryear{{de Blok}, {Walter}, {Brinks}, {Trachternach},
  {Oh}  \& {Kennicutt}}{{de Blok} et~al.}{2008}]{2008AJ....136.2648deBlok}
{de Blok} W.~J.~G.,  {Walter} F.,  {Brinks} E.,  {Trachternach} C.,  {Oh}
  S.~H.,   {Kennicutt} R.~C. J.,  2008, \mn@doi [\aj]
  {10.1088/0004-6256/136/6/2648}, \href
  {https://ui.adsabs.harvard.edu/abs/2008AJ....136.2648D} {136, 2648}

\bibitem[\protect\citeauthoryear{{del Peloso}, {da Silva}, {Porto de Mello}  \&
  {Arany-Prado}}{{del Peloso} et~al.}{2005}]{2005A&A...440.1153DelPelosoEtAl}
{del Peloso} E.~F.,  {da Silva} L.,  {Porto de Mello} G.~F.,   {Arany-Prado}
  L.~I.,  2005, \mn@doi [A\&A] {10.1051/0004-6361:20053307}, \href
  {https://ui.adsabs.harvard.edu/abs/2005A&A...440.1153D} {440, 1153}

\makeatother
\end{thebibliography}
	
	\newpage
	
	\appendix
	
	\section{Appendix}

	\subsection{Our model for the Mestel disc}
	
	We have left open the proofs of Lemma \ref{lemma velocity dispersion of model for Mestel disc} and Remark \ref{remark velocity max and min of model for Mestel disc}. For convenience let us write down Lemma \ref{lemma velocity dispersion of model for Mestel disc} again before we give it proof:

	\begin{customlem}{\ref{lemma velocity dispersion of model for Mestel disc}}
		When $\sigma\searrow 0$ the average tangential velocity in the model $f_0$ for the Mestel disc is
		\begin{equation*}
		v_{t,avg} = v_0 + o(\sigma)
		\end{equation*}
		independent of radius. The dispersion of the tangential velocities is
		\begin{equation*}
		\sigma + o(\sigma)
		\end{equation*}
		and also independent of radius; in the rest of this paper we refer to $\sigma$ as the velocity dispersion. Further the dispersion of the radial velocities is $\sqrt 2\sigma + o(\sigma)$ and
		\begin{equation} 
		C_0 = \frac{v_0^3}{8\sqrt 2 \pi^2 G \sigma^2} + o(\sigma^{-2}).
		\end{equation}
	\end{customlem}
	
	\begin{proof}[Proof of Lemma \ref{lemma velocity dispersion of model for Mestel disc}]
		All $o(\sigma^m)$ and $O(\sigma^m)$ terms in this proof are with respect to $\sigma\searrow 0$. Let $\sigma,v_0 > 0$. At radius $r>0$ the radial and the tangential component $v_r$ and $v_t$ of the velocity are distributed according to the law
		\begin{equation*}
		(v_r,v_t) \sim p(v_r,v_t) := \frac{1}{\Sigma_0(r)} f_0\left( r,0 , v_r,v_t \right)
		\end{equation*}
		where we have written $f_0$ in Cartesian coordinates; $x=(r,0)$ and $v=(v_1,v_2)=(v_r,v_t)$. $\int p\diff v = 1$ and we treat $p$ as a probability density function. Thank to \eqref{equ proof thm Mestel disc argument of Phi0} we have
		\begin{equation*}
		p(v_r,v_t) = \begin{cases}
		1/(Iv_t), & \text{if }(v_r^2+v_t^2-v_0^2)/2 - v_0^2 \log (v_t/v_0) \leq (2\sigma)^2 \text{ and } v_t\geq 0 \\
		0, & \text{else.}
		\end{cases}
		\end{equation*}
		We see that $p$ does not depend on the radius. The average tangential velocity is the expected value $\text E(v_t)$ and the square of the velocity dispersion is the Variance $\text{Var}(v_t)$. We are interested in the behaviour of $\text E(v_t)$ and $\text{Var}(v_t)$ when $\sigma\searrow 0$. Then the support of $p$ shrinks more and more. Since the point $(0,v_0)$ is always located inside the support of $p$ it is convenient to introduce coordinates that zoom onto that point while $\sigma\searrow 0$. We use the coordinate transformation
		\begin{align*}
		& v_r = \sigma w_r, \\
		& v_t = v_0 + \sigma ( w_t - v_0 ).
		%			& w_r := \frac{v_r}{\sigma}, \\
		%			& w_t := v_0 + \frac{1}{\sigma} (v_t - v_0).
		\end{align*}
		Then $\diff v = \sigma^2 \diff w$,
		\begin{equation*}
		(w_r, w_t) \sim q(w_r,w_t) = \sigma^2 p(\sigma w_r,v_0+\sigma(w_t-v_0))
		\end{equation*}
		and
		\begin{equation*}
		\text E(v_t) = \text E(v_0+\sigma(w_t-v_0)).
		\end{equation*}
		For $\sigma>0$ we denote the support of $q$ by
		\begin{align*}
		W_\sigma := &\left\{ (w_r,w_t) \middle| w_t\geq v_0-\frac{v_0}{\sigma}, \, \sigma^2\frac{w_r^2+(w_t-v_0)^2}{2} + \sigma v_0(w_t-v_0) - v_0^2\log\left(1+\sigma\frac{w_t-v_0}{v_0}\right) \leq 4\sigma^2 \right\} \\
		= &  \left\{ (w_r,w_t) \middle| w_t\geq v_0-\frac{v_0}{\sigma}, \, \frac{w_r^2+(w_t-v_0)^2}{2} +\frac{v_0(w_t-v_0)}{\sigma} -\frac{v_0^2}{\sigma^2} \log\left(1+\sigma\frac{w_t-v_0}{v_0}\right) \leq 4  \right\}.
		\end{align*}
		Thus
		\begin{equation*}
		q(w_r,w_t) = \begin{cases}
		\sigma^2/(I(v_0+\sigma(w_t-v_0))), & \text{if }(w_r,w_t)\in W_\sigma \\
		0, & \text{else.}
		\end{cases}
		\end{equation*}
		We set
		\begin{equation*}
		W_0 := \left\{ (w_r,w_t) \middle| \frac{w_r^2}{2} + (w_t-v_0)^2 \leq 4 \right\}.
		\end{equation*}
		Using that there is a compactum $K\subset\R^2$ such that
		\begin{equation*}
		W_\sigma\subset K
		\end{equation*}
		for all $\sigma>0$ sufficiently small, a second order Taylor approximation gives
		\begin{align*}
		\frac{v_0^2}{\sigma^2}\log\left(1+\sigma\frac{w_t-v_0}{v_0}\right) &= \frac{v_0^2}{\sigma^2}\left( \sigma\frac{w_t-v_0}{v_0} - \frac 12 \sigma^2 \frac{(w_t-v_0)^2}{v_0^2} + O(\sigma^3) \right) \\
		& = \frac{v_0(w_t-v_0)}{\sigma} - \frac 12 (w_t-v_0)^2 + O(\sigma).
		\end{align*}
		Thus for $\sigma\searrow 0$ the envelope of $W_\sigma$ converges uniformly to the envelope of $W_0$. This fact we will employ frequently below. We have
		\begin{align}
		\text E(v_t) & = \text E(v_0 + \sigma (w_t-v_0)) \nonumber\\
		& = \int(v_0+\sigma(w_t-v_0)) q(v_r,v_t) \diff w \nonumber \\
		&= \frac{\sigma^2}{I} \int_{W_\sigma} \diff w. \label{equ proof velocity dispersion E(vt)}
		\end{align}
		We have
		\begin{align*}
		I &= \int_{\{(v_r^2+v_t^2-v_0^2)/2 - v_0^2\log(v_t/v_0) \leq (2\sigma)^2,\,v_t\geq 0\}} \frac{\diff v}{v_t} = \sigma^2 \int_{W_\sigma} \frac{\diff w}{v_0+\sigma(w_t-v_0)}
		\end{align*}
		and hence
		\begin{equation} \label{equ proof velocity dispersion sigma square over I}
		\frac{\sigma^2}{I} \rightarrow \frac{v_0}{\measure(W_0)} \quad \text{for }\sigma\rightarrow 0.
		\end{equation}
		If we would pass to the limit in \eqref{equ proof velocity dispersion E(vt)} now, we would only get
		\begin{equation*}
		\text E(v_t) = v_0 +o(1).
		\end{equation*}
		This result makes only use of the fact that $(0,v_0)$ lies inside of $W_\sigma$ and  is too weak to analyse $\text{Var}(v_t)$. We must elaborate that $(0,v_0)$ marks the centre of $W_0$ to get better convergences:
		\begin{align*}
		\text E(v_t) & = v_0 + \frac{\sigma^2}{I}\int_{W_\sigma} \diff w - v_0 \frac{\sigma^2}{I} \int_{W_\sigma} \frac{\diff w}{v_0+\sigma(w_t-v_0)}\\
		& = v_0 + \frac{\sigma^2}{I}\int_{W_\sigma}\left(1-\frac{v_0}{v_0+\sigma(w_t-v_0)}\right)\diff w \\
		& = v_0 + \frac{\sigma^3}{I} \int_{W_\sigma}\frac{w_t-v_0}{v_0+\sigma(w_t-v_0)} \diff w .
		\end{align*}
		Using \eqref{equ proof velocity dispersion sigma square over I} and that
		\begin{equation*}
		\int_{W_\sigma}\frac{w_t-v_0}{v_0+\sigma(w_t-v_0)} \diff w \rightarrow \int_{W_0}\frac{w_t-v_0}{v_0} \diff w = 0 \quad \text{for }\sigma\searrow 0,
		\end{equation*}
		we get
		\begin{equation} \label{equ proof velocity dispersion E(vt) result}
		\text E(v_t) = v_0 + o(\sigma).
		\end{equation}
		We have
		\begin{align*}
		\text{Var} (v_t ) &= \int |v_t-\text E(v_t)|^2 p(v_r,v_t) \diff v \\
		& = \int | v_t - v_0 + v_0 - E(v_t)|^2 p(v_r,v_t) \diff v \\
		& = \int|v_t-v_0|^2p\diff v + 2\int(v_t-v_0)(v_0-\text E(v_t))p\diff v + |v_0-\text E(v_t)|^2.
		\end{align*}
		\eqref{equ proof velocity dispersion E(vt) result} and
		\begin{equation*}
		\|v_t-v_0\|_{L^\infty(\supp p)} = O(\sigma)
		\end{equation*}
		imply
		\begin{equation*}
		\text{Var}(v_t) = \int|v_t-v_0|^2p\diff v + o(\sigma^2).
		\end{equation*}
		Since 
		\begin{equation*}
		\int|v_t-v_0|^2p\diff v = \frac{\sigma^4}{I}\int_{W_\sigma}\frac{|w_t-v_0|^2}{v_0+\sigma(w_t-v_0)} \diff w
		\end{equation*}
		and
		\begin{equation*}
		\int_{W_\sigma}\frac{|w_t-v_0|^2}{v_0+\sigma(w_t-v_0)} \diff w \rightarrow \int_{W_0}\frac{|w_t-v_0|^2}{v_0} \diff w \quad \text{for }\sigma\searrow 0,
		\end{equation*}
		we have
		\begin{equation*}
		\int|v_t-v_0|^2p\diff v - \sigma^2\frac{\sigma^2}{I} \int_{W_0}\frac{|w_t-v_0|^2}{v_0} \diff w = o(\sigma^2)
		\end{equation*}
		and
		\begin{equation*}
		\text{Var}(v_t) = \sigma^2\frac{\sigma^2}{Iv_0}\int_{W_0}|w_t-v_0|^2\diff w + o(\sigma^2).
		\end{equation*}
		With the transformation $w_t-v_0 = 2s$
		\begin{align*}
		\int_{W_0}|w_t-v_0|^2\diff w & = \int_{v_0-2}^{v_0+2}\int_{-\sqrt{8-2(w_t-v_0)^2}}^{\sqrt{8-2(w_t-v_0)^2}}\diff w_r \, |w_t-v_0|^2 \diff w_t \\
		& = 2\sqrt 8 \int_{v_0-2}^{v_0+2}\sqrt{1-\frac{(w_t-v_0)^2}{4}}|w_t-v_0|^2\diff w _t \\
		& = 32\sqrt 2\int_{-1}^1\sqrt{1-s^2}s^2\diff s. %\\
		%& = 4\sqrt 2\pi = \measure(W_0).
		\end{align*}
		Using the transformation $s^2=t$, $\diff t = 2 s\diff s$ and the Beta-function $B$
		\begin{align*}
		\int_{W_0}|w_t-v_0|^2\diff w & = 64\sqrt 2 \int_0^1 \sqrt{1-s^2}s^2 \diff s = 32 \sqrt 2 \int_0^1 \sqrt{1-t} \sqrt t \diff t \\
		& = 32\sqrt 2 B(3/2,3/2) = 4\sqrt 2\pi = \measure(W_0).
		\end{align*}
		Since
		\begin{equation*}
		\frac{\sigma^2}{Iv_0} \rightarrow \measure(W_0)^{-1}\quad \text{for }\sigma \searrow 0,
		\end{equation*}
		this implies that
		\begin{equation*}
		\text{Var}(v_t) = \sigma^2 + o(\sigma^2).
		\end{equation*}
		Due to symmetry the average radial velocity is zero. $\text{Var}(v_r)$ can be calculated with the same techniques as above and one gets that the dispersion of the radial velocities is given by $\sqrt 2\sigma + o(\sigma)$. Further
		\begin{equation*}
		C_0 = \frac{v_0^2}{2\pi GI} = \frac{v_0^2}{2\pi G\sigma^2}\frac{\sigma^2}{I}
		\end{equation*} 
		and
		\begin{equation*}
		\left|C_0 - \frac{v_0^3}{8\sqrt 2 \pi^2 G\sigma^2}\right| = \frac{v_0^2}{2\pi G\sigma^2} \left| \frac{\sigma^2}{I} - \frac{v_0}{4\sqrt 2 \pi} \right|.
		\end{equation*}
		By \eqref{equ proof velocity dispersion sigma square over I}
		\begin{equation*}
		C_0 = \frac{v_0^3}{8\sqrt 2 \pi^2 G \sigma^2} + o(\sigma^{-2}).
		\end{equation*}

	\end{proof}

	In Remark \ref{remark velocity max and min of model for Mestel disc} we had left open the proof that the minimal and the maximal appearing tangential velocity in our self-consistent model $f_0$ for the Mestel disc are given by
	\begin{equation}
		v_{t,min/max} = v_0 \sqrt{-W_{0/-1}\left[ -\exp\left( -\frac{8\sigma^2}{v_0^2} - 1 \right) \right]}.
	\end{equation}
	Let us prove this, too.
	\begin{prop}
		Let $a \in \R$, then
		\begin{equation*}
			x^2+a = \log x, \quad x>0,
		\end{equation*}
		has
		\begin{enumerate}[label=]
			\item two solutions iff $a< - \log\sqrt 2 - 1/2$,
			\item one solution iff $a = - \log\sqrt 2 - 1/2$,
			\item no solution iff $a> - \log\sqrt 2 - 1/2$.
		\end{enumerate}
		The solutions are
		\begin{equation*}
			x_i = \sqrt{-\frac{1}{2}W_i(-2e^{2a})}
		\end{equation*}
		with $i=0,-1$ and $W_i$ denoting the $i$-th branch of the Lambert W function.
	\end{prop}

	\begin{rem} \label{remark Lambert W}
		The two Lambert W functions $W_{-1}:[-1/e,0)\rightarrow(-\infty,-1]$ and $W_0:[-1/e,\infty)\rightarrow[-1,\infty)$ are the two branches of the inverse function of
		\begin{equation*}
			(-\infty,\infty)\ni z\mapsto ze^z \in [-1/e,\infty).
		\end{equation*}
		In particular for all admissible $\zeta\geq -1/e$ and $i=0,-1$
		\begin{equation*}
			\zeta = W_i(\zeta) \exp(W_i(\zeta)).
		\end{equation*}
	\end{rem}

	\begin{proof}[Proof of the Proposition]
		Use the transformation
		\begin{equation*}
			x = \exp\left(-\frac{y}{2} + a\right) , \quad y\in\R.
		\end{equation*}
		Then
		\begin{align} \label{equ proof prop Lambert W solution x iff solution y}
			x^2 + a = \log x &\iff \exp\left(-y + 2a\right) = -\frac{y}{2} \nonumber \\
			& \iff -2e^{2a} = y e^y.
		\end{align}
		Since
		\begin{equation*}
			(ye^y)' = (y+1)e^y = 0 \iff y = -1
		\end{equation*}
		we see that $ye^y$ takes its global minimum at $y=-1$ and hence
		\begin{equation*}
			ye^y \geq -\frac{1}{e}
		\end{equation*}
		for all $y\in\R$. Thus \eqref{equ proof prop Lambert W solution x iff solution y} can hold iff
		\begin{equation*}
			2e^{2a} \leq \frac{1}{e} \iff a \leq -\log\sqrt 2 - \frac 12.
		\end{equation*}
		In this case we have the two solutions
		\begin{equation*}
			y_i = W_i(-2e^{2a})
		\end{equation*}
		with $i=0$ or $i=-1$. Thus
		\begin{equation*}
			x_i = \exp\left(-\frac{y_i}{2} + a \right) = \exp\left( - \frac 12 W_i(-2e^{2a})+a\right)
		\end{equation*}
		are the solutions of $x^2+a=\log x$. We want to simplify this formula further. For this purpose let $b\in\R$. With $\zeta = -e^b$ we get as in Remark \ref{remark Lambert W}
		\begin{equation*}
			-e^b = W_i(-e^b) \exp(W_i(-e^b)).
		\end{equation*}
		Thus
		\begin{equation*}
			\frac{-1}{W_i(-e^b)} = \exp(W_i(-e^b) - b)
		\end{equation*}
		Taking the logarithm on both sides of the equation gives
		\begin{equation*}
		-\log(-W_i(-e^b)) = W_i(-e^b) - b, \quad b\in\R.
		\end{equation*}
		Applying this equality with $b=2a+\log 2$ gives
		\begin{align*}
			\exp\left( - \frac 12 W_i(-2e^{2a})+a\right) &= \exp\left[ -\frac{1}{2} \left[ W_i(-e^{2a + \log 2}) - (2a + \log 2) \right] \right] \exp\left[ -\frac 12 \log 2 \right]\\
			& = \frac 1{\sqrt 2} \exp \left[ \frac{1}{2} \log\left( -W_i(-e^{2a+\log 2}) \right) \right] \\
			& = \sqrt{-\frac 12 W_i(-2e^{2a})}.
		\end{align*}
	\end{proof}

	\begin{lem}
		For $v_0,\sigma>0$ and $f_0$ as in Lemma \ref{lemma velocity dispersion of model for Mestel disc} the minimal and the maximal appearing tangential velocity at every position $x\in\R^2\backslash\{0\}$ are given by
		\begin{equation*}
			v_{t,min/max} = v_0 \sqrt{-W_{0/-1}\left[ -\exp\left( -\frac{8\sigma^2}{v_0^2} - 1 \right) \right]}
		\end{equation*}
	\end{lem}

	\begin{proof}
		Denote by $v_r,v_t$ the radial and the tangential component of the velocity. For every $x\in\R^2\backslash\{0\}$
		\begin{align*}
			(v_r,v_t) \in \supp f_0(x,\cdot) & \iff \frac{v_r^2+v_t^2}{2} - v_0^2\log \frac{v_t}{v_0} - \frac{v_0^2}{2} \leq (2\sigma)^2 \\
			& \iff \frac{1}{2} \left(\frac{v_t}{v_0}\right)^2 - \log \frac{v_t}{v_0} - \frac 12 \leq \left(\frac{2\sigma}{v_0}\right)^2 - \frac 12 \left(\frac{v_r}{v_0}\right)^2.
		\end{align*}
		Since the left side of the last inequality is strictly convex and diverges to infinity for $v_t\rightarrow 0$ and $v_t\rightarrow\infty$, the minimal and maximal tangential velocity are the two solutions of
		\begin{equation*}
			\frac{1}{2} \left(\frac{v_t}{v_0}\right)^2 - \log \frac{v_t}{v_0} - \frac 12 = \left(\frac{2\sigma}{v_0}\right)^2
		\end{equation*}
		With 
		\begin{equation*}
			x^2 = \frac 12 \left(\frac{v_t}{v_0} \right)^2
		\end{equation*}
		this holds iff
		\begin{align*}
			& x^2 - \log \sqrt 2 - \log x - \frac 12 = \left(\frac{2\sigma}{v_0}\right)^2 \\
			&\iff  x^2 - \left(\log \sqrt 2 + \frac 12 + \left(\frac{2\sigma}{v_0}\right)^2 \right) = \log x.
		\end{align*}
		This has the two solutions
		\begin{align*}
			x_i &= \sqrt{ -\frac 12 W_i\left[ -2\exp\left( -2\log \sqrt 2 - 1 - \frac{8\sigma^2}{v_0^2} \right) \right] } \\
			&= \sqrt{-\frac 12 W_i\left[ -\exp\left( -\frac{8\sigma^2}{v_0^2} - 1 \right) \right]}.
		\end{align*}
		Thus
		\begin{equation*}
			v_{t,min} = \sqrt 2 v_0 x_0 = v_0 \sqrt{-W_0\left[ -\exp\left( -\frac{8\sigma^2}{v_0^2} - 1 \right) \right]}
		\end{equation*}
		and
		\begin{equation*}
		v_{t,max} = \sqrt 2 v_0 x_{-1} = v_0 \sqrt{-W_{-1}\left[ -\exp\left( -\frac{8\sigma^2}{v_0^2} - 1 \right) \right]}
		\end{equation*}
		
	\end{proof}

	\subsection{Average $z$-distance in a disk with constant scale height}

	We claimed in Section \ref{subsection-cutout-Mestel-disc} that for a disc with the spatial density
	\begin{equation*}
	\Sigma(x_1,x_2)\exp\left( -|z|/z_g \right)
	\end{equation*}
	the expected value of the distance in $z$-direction between two particles that we draw at random from this density is $1.5z_g$. We give a short derivation of this. W.l.g. we set $z_g=1$ in the following calculations. Then the $z$-coordinates $Z_1$ and $Z_2$ of the two particles are distributed according to the law
	\begin{equation*}
	Z_1,Z_2 \sim \frac{1}{2} e^{-|z|}.
	\end{equation*}
	The probability that the distance between the two particles is lower than $\delta>0$ is
	\begin{equation*}
	P(|Z_1-Z_2|<\delta) = \frac{1}{4} \int_{-\infty}^{\infty} \int_{z-\delta}^{z+\delta} e^{-|z|}e^{-|z'|} \diff z' \diff z.
	\end{equation*}
	The probability distribution function $p(\delta)$ that corresponds to the random variable $|Z_1-Z_2|$ is given by
	\begin{align*}
	p(\delta) & = \frac{\diff}{\diff \delta} P(|Z_1-Z_2|<\delta) \\
	& = \frac{1}{4} \int_{-\infty}^\infty e^{-|z|}\left( e^{-|z+\delta|} + e^{-|z-\delta|} \right) \diff z \\
	& = \frac 12 \int_{-\infty}^\infty e^{-|z|}e^{-|z+\delta|} \diff z \\
	& = \frac 12 \int_{-\infty}^{-\delta} e^ze^{z+\delta} \diff z + \frac 12 \int_{-\delta}^0 e^ze^{-z-\delta} \diff z + \frac 12 \int_0^\infty e^{-z}e^{-z-\delta} \diff z \\
	& = \frac{e^\delta}{2}\int_{-\infty}^{-\delta} e^{2z}\diff z + \frac{e^{-\delta}}{2} \int_{-\delta}^0 \diff z + \frac{e^{-\delta}}{2} \int_0^\infty e^{-2z} \diff z \\
	& = \frac{e^\delta e^{-2\delta}}{4} + \frac{\delta e^{-\delta}}{2} + \frac{e^{-\delta}}{4} \\
	& = \frac 12(1+\delta) e^{-\delta}
	\end{align*}
	if $\delta>0$, and $p(\delta)=0$ if $\delta \leq 0$. Since
	\begin{equation*}
	\int_0^\infty \delta e^{-\delta} \diff \delta = \int_0^\infty\delta (-e^{-\delta})' \diff \delta = \int_0^\infty e^{-\delta} \diff \delta = 1
	\end{equation*}
	and
	\begin{equation*}
	\int_0^\infty \delta^2 e^{-\delta} \diff \delta = 2 \int_0^\infty \delta e^{-\delta} \diff \delta = 2,
	\end{equation*}
	we have for the expected value
	\begin{align*}
	EV(|Z_1-Z_2|) & = \int_0^\infty \frac{\delta}{2} (1+\delta) e^{-\delta}\diff \delta \\
	& = \frac 12 \int_0^\infty \delta e^{-\delta} \diff \delta + \frac 12 \int_0^\infty \delta^2 e^{-\delta} \diff \delta = \frac{3}{2}.
	\end{align*}
	These calculations show that for a disc with scale height $z_g$ the distance in $z$-direction between two randomly determined particles is at average $1.5z_g$.

	\subsection{Integrating the equations of motions in Section \ref{section-Stability-and-Spiral-Activity}}
	
	In Section \ref{section-Stability-and-Spiral-Activity} we discuss several simulations. We want to describe the numerical methods used in these simulations. 
	
	For every simulation we use a distribution function $f(x,v)$ that was generated with the algorithm from Section \ref{subsection-algorithm}. From this distribution function we draw at random particles with initial coordinates $(x_i(0),v_i(0))$ and integrate the equations of motion
	\begin{equation*}
	\dot{x_i} = v_i,\quad \dot{v_i} = G \sum_{j\neq i} \frac{x_j-x_i}{\left(|x_j-x_i|^2+\delta_z^2\right)^{3/2}}.
	\end{equation*}
	We have modified Newton's law of gravitation to take into account the disc thickness as we have done already in Section \ref{subsection-cutout-Mestel-disc}. As previously $\delta_z = 1.5z_g$ with $z_g=\SI{300}{\pc}$.
	
	We integrate the equations of motion numerically using a Velocity Verlet algorithm and we have implemented two ways to calculate the force efficiently. The first method uses a polar grid where the disc is divided in 100 uniform angles and along each radial line 400 meshpoints are distributed out to \SI{100}{\kpc}. Particles that move beyond \SI{100}{\kpc} are dropped from the simulation. To distribute the meshpoints radially, we approximate the density of our initial data by a continuous, piecewise function which is linearly increasing to \SI{4.3}{\kpc} and exponentially decreasing beyond. The meshpoints are distributed such that to each meshpoint the same mass of the approximate density would be assigned. In the simulation we use bi-linear interpolation to assign the mass of each particle to the four adjacent meshpoints, calculate the force between the meshpoints and get the force on the particles by again using bi-linear interpolation. For simulations with this polar mesh, we created 1 Million particles from $f(x,v)$. The second method enforces axial symmetry. This is achieved by calculating in each time step a histogram of the radial positions of the particles. This histogram is transformed into an axially symmetric density and from this density the forces on the particles are calculated. For simulations with enforced axial symmetry we used 100.000 particles.
	
\end{document}